\documentclass{amsart}

\usepackage{a4wide}
\usepackage{amsmath,amsfonts,amssymb}
\usepackage[pdftex]{graphicx}
\usepackage{mathptmx}

\usepackage{xspace}
\usepackage{xfrac}

\usepackage{subfig}

\usepackage{natbib}

\newcommand{\xs}{x^*}
\newcommand{\bx}{\bar{x}}
\newcommand{\bP}{\bar{\Psi}}
\newcommand{\A}{\ensuremath{\mathbb{A}}\xspace}
\newcommand{\B}{\ensuremath{\mathbb{B}}\xspace}
\newcommand{\R}{\mathbb{R}}
\newcommand{\rd}{\mathrm{d}}
\newcommand{\bo}[1]{\mathrm{O}\left(#1\right)}
\newcommand{\lo}[1]{\mathrm{o}\left(#1\right)}
\newcommand{\F}{\mathcal{F}}
\newcommand{\G}{\mathcal{G}}
\newcommand{\C}{\mathcal{C}}
\newcommand{\I}{\mathcal{I}}
\newcommand{\bydef}{:=}
\newcommand{\eps}{\varepsilon}
\newcommand{\e}{\mathrm{e}}
\newcommand{\N}{\mathcal{N}}
\newcommand{\RR}{\mathcal{R}}
\newcommand{\HH}{\mathcal{H}}
\newcommand{\ord}{\mathrm{ord}}
\newcommand{\QQ}{\mathcal{Q}}
\newcommand{\D}{\mathcal{D}}
\newcommand{\tK}{\tilde{K}}
\newcommand{\tx}{\tilde{x}}

\newcommand{\st}{s^*}
\newcommand{\hs}{\hat{s}}
\newcommand{\bs}{\bar{s}}
\newcommand{\tC}{\tilde{C}}

\newcommand{\HHH}{\mathfrak{H}}
\newcommand{\RRR}{\mathfrak{R}}

\newtheorem{remark}{Remark}
\newtheorem{definition}{Definition}
\newtheorem{theorem}{Theorem}
\newtheorem{proposition}{Proposition}
\newtheorem*{lemma}{Lemma}

\graphicspath{{./Figures/}}

\usepackage[utf8]{inputenc}

\begin{document}

\title{Fixation in large populations: a continuous view of a discrete problem}

\date{\today}

%
%

\author{Fabio A. C. C. Chalub}

\address{Departamento de Matem\'atica and Centro de Matem\'atica e Aplica\c c\~oes, Universidade Nova de Lisboa, Quinta da Torre, 2829-516, Caparica, Portugal.}

\author{Max O. Souza}
\address{Departamento de Matem\'atica Aplicada, Universidade Federal Fluminense, R. M\'ario Santos Braga, s/n, 22240-920, Niter\'oi, RJ, Brasil.}
\email{msouza@mat.uff.br}

\thanks{
FACCC was partially supported by the FCT/Portugal through the project PEst-OE/MAT/UI0297/2014 and by an Investigador FCT grant. MOS was partially supported by CNPq under grants \#~309616/2009-3 and  \#~308113/2012-8.
The authors thank two anonymous reviewers and the associate editor for several remarks that greatly improved the paper.
}



%
%

\begin{abstract}
We study fixation in large, but finite, populations with two types, and dynamics governed  by birth-death processes. By considering a restricted class of  such processes, which includes many of the   evolutionary processes usually discussed in the literature,  we derive a continuous approximation for the probability of fixation that is valid beyond the weak-selection (WS)  limit. Indeed, in the derivation three regimes naturally appear: selection-driven, balanced, and quasi-neutral --- the latter two require WS, while the former can appear with or without WS. From the continuous approximations, we then   obtain asymptotic approximations for evolutionary dynamics with at most one equilibrium, in the selection-driven regime, that does not preclude a weak-selection regime. As an application, we study the fixation pattern when the infinite population limit has an interior Evolutionary Stable Strategy (ESS): (i) we show that the fixation pattern for the Hawk and Dove game satisfies what we term  the one-half law: if the Evolutionary Stable Strategy (ESS) is outside a small interval around $\sfrac{1}{2}$, the fixation is of dominance type; (ii) we also show that, outside of the weak-selection regime, the long-term dynamics of large populations can have very little resemblance to the infinite population case; in addition, we also present results for the case of two equilibria, and show that even when there is weak-selection the long-term dynamics can be dramatically different from the one predicted by the Replicator Dynamics. Finally, we present continuous restatements  valid for large populations of two classical concepts naturally defined in the discrete case: (i) the definition of an $\textsc{ESS}_N$ strategy; (ii) the definition of a risk-dominant strategy. We then present two applications of these restatements: (i) we obtain an asymptotic definition valid in the quasi-neutral regime that recovers both the one-third law under linear fitness and the generalised one-third law for $d$-player games; (ii) we extend  the ideas behind the (generalised) one-third law  outside the quasi-neutral regime and, as a generalisation, we introduce the concept of critical-frequency; (iii) we recover the classification of risk-dominant strategies for $d$-player games.
\end{abstract}

\maketitle

\section{Introduction}

\subsection{Background}

One of the most natural questions addressed in the study of evolution is the long term behaviour of types distribution.  In an number of settings, it is known that only one type will be present at sufficiently long times. This is known as fixation, and its likelihood  is known as the fixation probability of a given type. This  has been studied since the works by  \citet{Wright_1931}, \citet{Fisher_1930}, \citet{Moran} and \citet{Kimura} for the case of neutral evolution, and frequency independent fitness,  both in discrete and continuous settings~\citep{Nowak:06,Ewens_2004}.

An alternative approach was also taken by game-theorists who, in some appropriate sense, were developing a mathematical description of the Darwinian theory of evolution~\citep{Smith_1982}. Among the models studied, possibly the most popular is the Replicator Dynamics, which assumes an  infinite and well-mixed population~\citep{TaylorJonker_1978}.  In this model, unlike in the finite population case, it is possible to find mixed stable populations, and this leads to the development of what is currently known as an Evolutionary Stable Strategy (ESS)~\citep{HofbauerSigmund1998}. 

Historically, the blending of these approaches began with the use of diffusion approximations to obtain continuous frequency-independent models that are  valid in the large population limit~\citep{Kimura,Feller1951}. See~\citet{Gillespie_Feldman} for a critical view on the use of diffusion approximations. Later on, we have the formulation of evolutionary dynamics in finite populations with frequency dependent fitness~\citep{TaylorFudenbergSasakiNowak,NowakSasakiTaylorFudenberg}; see also~\citet{Nowak:06}. A notable exception of this dichotomy is \citet{EthierKurtz}, where the formulation of a frequency dependent version of the Wright-Fisher process, together with its continuous limit under what is presently known as the weak-selection regime, cf.~~\citet{Ewens_2004},  are already discussed.

 More recently, a number of different studies attempted to blend the infinite population ideas of evolutionary game-theory to the randomness of stochastic finite population models using different approaches, both in terms of the modelling assumptions  and of mathematical rigour~\citep{ChalubSouza09b,ChalubSouza_JMB,Traulsen_etal_2005,Lessard:Ladret:2007,Lessard_2005, McKaneWaxman07, Waxman2011,ChampagnatFerriereMeleard_TPB2006,ChampagnatFerriereMeleard_SM2008,Traulsenetal2006,Fournier_Meleard_AAP2004,TraulsenClaussenHauert_PRE2012}.
 
Fixation probability  in finite  large populations, under the assumption of a mean-field approximation has been studied by  \citet{Kimura} and by  \citet{Gillespie:1981}. More recently, without such assumption, it has also    been studied   by  
 \citet{AntalScheuring_2006} with a focus on the invasion coefficients. It was also studied in \citet{Traulsenetal2006,Traulsenetal2006b} as part of more general studies on stochastic invasion and fixation, and evolutionary stability and also under the weak-selection regime by \cite{AltrockTraulsen2009a}, who studied the Fermi process and the frequency dependent Moran process,  but with a focus on fixation time rather than fixation probabilities. Further studies are \citet{AssafMobilia2010}  which tackled the study of fixation for coordination and co-existence situations using large-deviation theory, although for the latter it had more attention to approximation of quasi-stationary probability distribution, and  a  study in on random fluctuations about the meta-stable state  in    \citet{MobiliaAssaf2010}. All these studies tackle the problem directly from the discrete description. For a variation on the Fermi process that might  become deterministic at finite population size, see \citet{AltrockTraulsen2009b}.
 
The aim of this work is to further contribute to these studies by obtaining results regarding the fixation probability  that can make use of continuous approximations that are valid for large populations, but that do not require an infinite population limit to be valid. In particular, no weak-selection assumption will be necessary for deriving these approximations. This procedure leads to results that highlight how large, but finite, populations  can have fixation patterns that are very peculiar to this regime, while still having some connection with both finite and infinite frameworks.

\begin{figure}[htbp]
\begin{center}
	\includegraphics[scale=1]{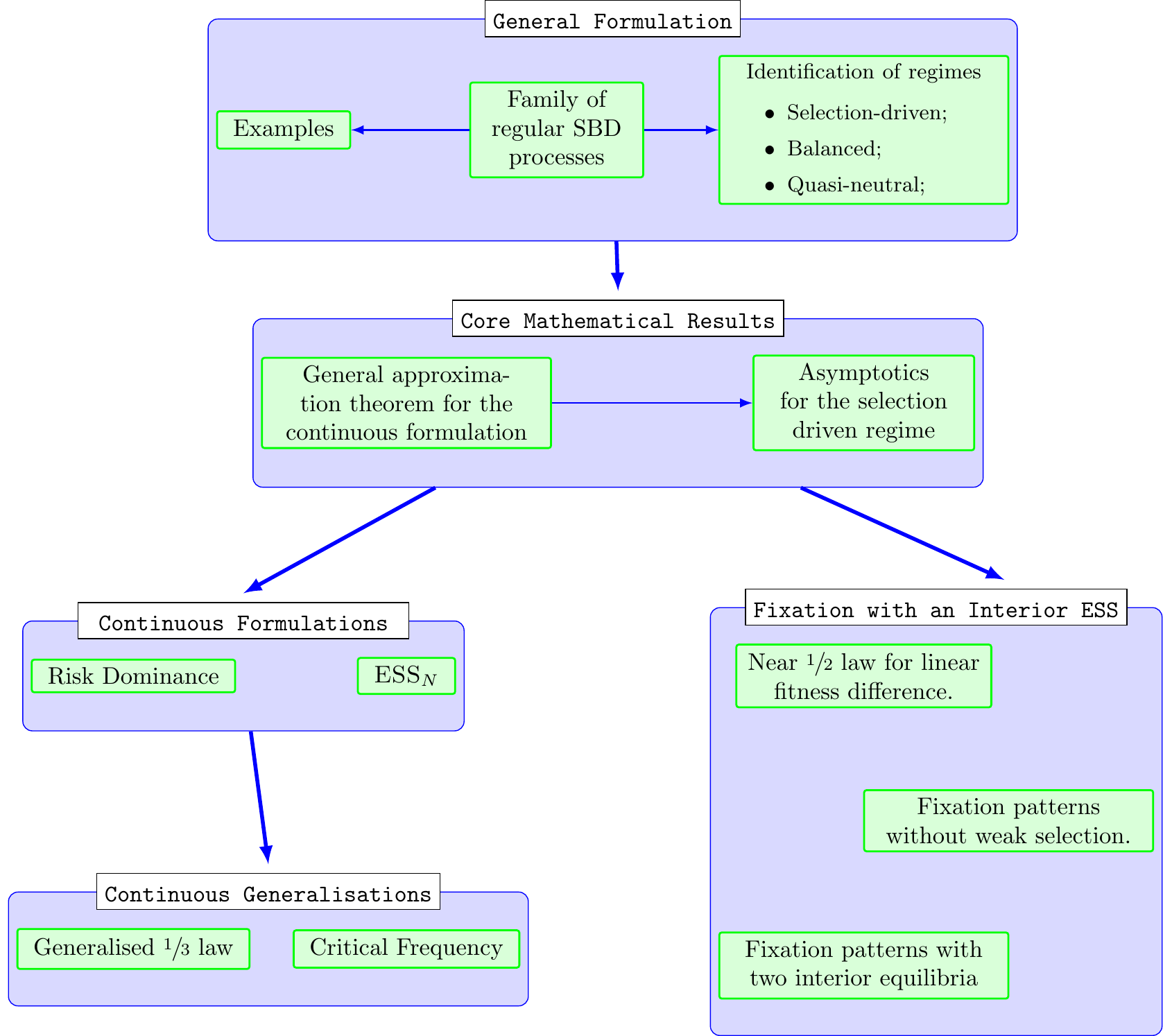}
\end{center}
\caption{A graphical roadmap of the results presented here. They are divided in four categories with the logical connection also displayed.}
\label{fig:rdmp}
\end{figure}

\subsection{Outline}

In Section~\ref{sec:gen_for}, we provide the general formulation for the processes we shall address in this work. We term such processes Suitable Birth Death (SBD) processes. In Section~\ref{sec:setup}, we  define what we call a regular family of SBD processes. For such processes, we justify the use of a diffusive approximation, valid for large,  but finite, size $N$. Such an approximation does not require the weak-selection limit to be derived, but it relies instead on the smoothness of the transition probabilities which, in turn, depends on the corresponding fitness functions. The qualitative nature of the fixation probabilities will depend not only on the functional form of fitness functions themselves, but also on what we denote the intensity of selection. Depending on the behaviour of the intensity of selection as $N$ goes to infinity, we identify three natural regimes: selection-driven, balanced, and  quasi-neutral; the last two can only appear in the so-called weak-selection regime. 

We then proceed to study the fixation probabilities in the selection-driven regime in Section~\ref{sec:sdr_asymp}, with the assumption that there is at most one neutral point (a point in which both fitness are equal), and then we necessarily have either dominance, coexistence or coordination. For the dominance case we recover a number of formulae in the literature in a unified framework. For coexistence we obtain a novel formula. For coordination, we obtain a generalisation, in terms of range of validity, of previously obtained formulae \citet{MobiliaAssaf2010,AssafMobilia2010}.

In section~\ref{sec:fde}, we show using the coexistence formulae derived in the previous section, that the fixation in the presence of a deterministic ESS can be quite distinguished from what would be expected on grounds of a deterministic approach. In the case of weak-selection, and linear log-differences at leading order, we obtain what we call the near one-half law: if the deterministic ESS is located exactly at one-half, then the fixation probability  is essentially constant at one-half. However, if the ESS equilibrium is outside of a small region --- in a precise sense --- around one-half then one has dominance by one of the types. Moreover, in the case of absence of weak-selection, we show examples of how equivalent games can lead to completely different fixation patterns. In particular, one can have a  frequency possibly close to one of a particular type at the deterministic ESS, and nevertheless one might also have  almost certain extinction of this type. In addition, we also briefly discuss fixation in the presence of two interior equilibria. The results here are also new and show that, in the presence of an interior ESS, the long-term dynamics of large, but finite population, can be markedly distinct  from the infinite population case, when differential equations are used. 

In section~\ref{sec:ess_lp}, we discuss how to use the continuous approximation derived in section~\ref{sec:setup} to obtain a continuous formulations valid for large populations of the celebrated $\mathrm{ESS}_N$ concept.  We then derive an asymptotic approximation in the quasi-neutral regime, that can be used to obtain a simplified condition for the existence of an $\mathrm{ESS}_N$ in a finite, large population. As special cases of such condition, we obtain the celebrated one-third law when the leading-order log-differences of the fitness is linear~\citep{NowakSasakiTaylorFudenberg} and the generalised one-third law for $d$-player games \citep{Gokhale:Traulsen:2010,Lessard:2011}. We also obtain a continuous formulation of a risk-dominant strategy and recover a result of \citet{Kurokawa:Ihara:2009}. We then study the existence of $\textsc{ESS}_N$ in more general contexts, and obtain what we call the critical frequency curves for the coordination case. In the case of selection driven regimes, we show how the asymptotic formulae can be used to simplify the computation of such curves.

Finally, a discussion of the results is presented in section~\ref{sec:discuss}. A graphical roadmap of the results is displayed in Figure~\ref{fig:rdmp}.

\section{General formulation}

\label{sec:gen_for}

We shall consider birth-death processes in a population of size $N$ with two types $\A$ and $\B$.
For $a,b$ such that $b-a\in N^{-1}\mathbb{N}\cup\{0\}$, we define
\[
 [a,b]_N\bydef\left\{a,a+\frac{1}{N},a+\frac{2}{N},\dots,b\right\}\ .
\]
In particular, we denote the fraction of type $\A$ individuals in the population by 
\[
x \in [0,1]_N=\left\{0,\frac{1}{N},\frac{2}{N},\ldots,1\right\}.
\]
The transition probabilities $1\ge T^+_N(x),T^0_N(x),T^-_N(x)\ge0$ indicate the probability that in a population with $xN$ type \A individuals, the next generation has $xN+1$ ($xN$, $xN-1$, respectively) type \A individuals and are given by
\begin{align*}
& T_N^\pm(x)=x(1-x)\Delta^\pm\left(\Psi_N^\A,\Psi_N^\B\right).\\
& T_N^0(x)=1-T^+_N(x)-T^-_N(x)\ .
\end{align*}
The factors $\Delta^\pm:\R^+\times\R^+\to\R^+$ model the natural selection.
Fitnesses are given by $\Psi_N^\A,\Psi_N^\B:[0,1]\to\R^+.$ By extension, we also write $\Delta^\pm_N(x)=\Delta^\pm\left(\Psi_N^\A(x),\Psi_N^\B(x)\right)$, where the $+$ ($-$) symbol indicates an increase (decrease, respectively) of the quantity of the \A individuals.

A birth-death process with these properties will be denoted a Suitable Birth-Death process or a SBD process for short. A similar class of processes has been studied by \citet{AssafMobilia2010,MobiliaAssaf2010}.

Among the many models contemplated in this set-up, the most common ones  are:
\begin{description}
\item[Frequency dependent Moran process~\citep{NowakSasakiTaylorFudenberg}]
\[
\Delta^+\left(\Psi_N^\A,\Psi_N^\B\right)=\frac{\Psi_N^\A}{\bP_{N}},\quad \Delta^-\left(\Psi_N^\A,\Psi_N^\B\right)=\frac{\Psi_N^\B}{\bP_{N}},\quad \bP_{N}(x)=x\Psi_N^\A(x)+(1-x)\Psi_N^\B(x).
\]
\item[Linear Moran process~\citep{TraulsenClaussenHauert_PRE2006}]
\[
\Delta^+\left(\Psi_N^\A,\Psi_N^\B\right)=\frac{1}{2}\left[1+\Psi_N^\A-\bP_{N}\right],\quad \Delta^-\left(\Psi_N^\A,\Psi_N^\B\right)=\frac{1}{2}\left[1+\Psi_N^\B-\bP_{N}\right]
\]
\item[Local update rule~\citep{TraulsenClaussenHauert_PRE2006}]
\[
\Delta^+\left(\Psi_N^\A,\Psi_N^\B\right)=\frac{1}{2}\left[1+\Psi_N^\A-\Psi_N^\B\right],\quad \Delta^-\left(\Psi_N^\A,\Psi_N^\B\right)=\frac{1}{2}\left[1+\Psi_N^\B-\Psi_N^\A\right]
\]
\item[Fermi process~\citep{Szabo_Hauert,AltrockTraulsen2009a}]
\[
\Delta^+\left(\Psi_N^\A,\Psi_N^\B\right)=\left(1+\exp\left(\Psi_N^\B-\Psi_N^\A\right)\right)^{-1},\quad 
\Delta^+\left(\Psi_N^\A,\Psi_N^\B\right)=\left(1+\exp\left(\Psi_N^\A-\Psi_N^\B\right)\right)^{-1}
\]
\end{description}

We also will make extensively use  of the $\mathcal{O}$ and $o$ notations. Recall that
\[
f(x)=\bo{g_1(x),\ldots,g_n(x)},\quad x\to x_0,
\]
if there are positive numbers $C$ and $\delta$ such that
\[
|f(x)|\leq C\max_{i=1,\ldots,n}\{|g_1(x)|,\ldots,|g_n(x)|\},\quad |x-x_0|<\delta.
\]
We also say that
\[
f(x)=o(g(x)), \quad x\to x_0,
\]
if
\[
\lim_{x\to x_0}\frac{f(x)}{g(x)}=0,
\]
provided that $g(x)\not=0$, when $x$ is close to $x_0$, but $x\not=x_0$.

Finally, we will say that
\[
f(x)=\ord(g(x)),\quad x\to x_0
\]
if there are constants $C_1,C_2>0$ such that
\[
C_1|g(x)|\leq f(x)\leq C_2|g(x)|,\quad x\to x_0.
\]
\section{Approximations for the fixation probabilities in large populations}

\label{sec:setup}

\subsection{A general continuous approximation}

We consider SBD processes, with fitness functions given by $\Psi^{\A}_N,\Psi^{\B}_N:[0,1]\to\R^+$.  The fixation probability is then  given by~\citep{AntalScheuring_2006}
\begin{equation}
\Phi_N(x)\bydef c^{-1}_N\sum_{s\in[\sfrac{1}{N},x]_N}\prod_{r\in[\sfrac{1}{N},s-\sfrac{1}{N}]_N}\frac{\Delta^{-}_N(r)}{\Delta^+_N(r)},
\label{eqn:fix_disc}
\end{equation}
with $c_N$ chosen such that $\Phi_N(1)=1$.

As observed above, a number of approximations for large $N$ to equation \eqref{eqn:fix_disc} have been obtained previously, in different regimes.  It turns out that we can obtain such an approximation, that is valid for  a number of different regimes, at the expense of requiring some extra regularity in the log-difference of the fitness part of the process. In order to this, we introduce a number of definitions. 
\begin{definition}[Generalised log relative fitness]
We define the \textit{generalised log difference of fitness} as
\[
\Theta_N(x)\bydef\log\left(\frac{\Delta^+_N(x)}{\Delta^-_N(x)}\right).
\]
Assume that
\[
\lim_{N\to\infty}\|\Theta_N\|_\infty=\zeta.
\]
If $\zeta=0$, we shall say that the evolutionary dynamics satisfies weak-selection, and if $0<\zeta\ll1$, we shall say that it satisfies moderate selection.

\end{definition}
\begin{remark}
Weak selection means that both types perform nearly as well in large population. Nevertheless, a precise quantification of ``nearly'' is needed in order to determine the resulting dynamics. In particular, different scalings for such decay can lead to very different dynamics. 
For example, if $\Psi^{(\A,\B)}(x)=1+w(N)\varphi^{(\A,\B)}_N(x)$, with $w(N)\ll1$ and $0<\varepsilon<\|\varphi^{(\A)}_N-\varphi^{(\B)}_N\|_\infty<M$ then $\lim_{N\to\infty}\|\Theta_N\|_\infty=0$ if and only if $\lim_{N\to\infty}w(N)=0$. This is consistent with the classical weak-selection assumption; see~\citet{Nowak:06}. Note that the use of the supremum norm indicates that \A and \B perform nearly as well in all possible scenarios.
See also \citet{ChalubSouza09b,ChalubSouza_JMB} for discussion on the different possible scalings.
\end{remark}

\begin{definition}[Formal infinite population limit]
\label{def:taf}
We say that a family, indexed by population size, of frequency dependent Suitable Birth-Death processes with fitness functions $\Psi^{\A,\B}_N$ has a formal infinite population limit, if
\begin{enumerate}
\item There is $M>0$, such that $0<\|\Theta_N\|_\infty<M$;
\item\label{def:frm_b} There exists $\theta\in C^0([0,1])$,    with $\|\theta\|_\infty=1$ such that
\[
\lim_{N\to\infty} \epsilon_N=0,\quad \epsilon_N=\left\|\frac{\Theta_N}{\|\Theta_N\|_\infty}-\theta\right\|_\infty;
\]
\item $\theta$ has finitely many zeros.
\end{enumerate}

In this case, by extension, we also say that $\theta$ is the formal infinite population limit of $\Theta_N$, and we define the fitness potential as 
\[
\F(s)=-\int_0^s\theta(r)\,\rd r.
\]
We shall also say  that $\F$ is an interior potential if its global maximum is only attained at the interior; otherwise, we will say that it is a boundary potential.
\end{definition}

\begin{definition}[Family of regular SBD]
\label{def:frm}
Consider a family, indexed by population size, of frequency dependent Suitable Birth-Death processes with fitness functions $\Psi^{\A,\B}_N$.
We shall say that such a family  is regular, if 
\begin{enumerate}
\item $\Theta_N$ is $C^1$and it  has a formal infinite population limit $\theta\in C^2([0,1])$.
\item Let $\epsilon_N$ be as in definition~\ref{def:taf}, and 
\[
\kappa_N^{-1}=N\|\Theta_N\|_\infty.
\]
If
\[
\lim_{N\to\infty}\kappa^{-1}_N=\infty,
\]
then we also require that
\[
\lim_{N\to\infty}\kappa^{-1}_N\epsilon_N=0.
\]
\end{enumerate}
 \end{definition}
\begin{remark}
Definition~\ref{def:frm} includes weak-selection, size-independent fitness functions among many other settings, provided that the underlying fitness functions are sufficient regular, which turns out to be satisfied by most applications and examples in the literature. Notice also that this contrasts with less stringent assumptions used to justify a diffusion limit---cf.~\citet{ChalubSouza09b,ChalubSouza_JMB}
\end{remark}

For a family of regular SBD processes, we can approximate equation \eqref{eqn:fix_disc} as follows:
\begin{theorem}
\label{thm:cont_approx}
Assume that we have a regular family of SBD processes, and assume further that the formal infinite population limit, $\theta$, does not vanish at the boundaries. 

Then, for sufficient large $N$,  the fixation probability can be approximated as follows:
\begin{equation}
\Phi_N(x)=\phi_N(x)+\bo{\kappa^{-1}_N\epsilon_N,\kappa_N\xi_N^2,\kappa_N^{1-b}\xi_N^2},
\label{eqn:fit_is}
\end{equation}
where $\xi_N=\|\Theta_N\|_\infty$, $b=1$ if $\F$ is a boundary potential, and $b=0$ otherwise, and
\begin{equation}
\label{eqn:asymp_frm}
 \phi_N(x)=d^{-1}_N\int_0^x\exp\left(\kappa_N^{-1}\F(s)\right)\, \rd s,\quad d_N=\int_0^1\exp\left(\kappa_N^{-1}\F(s)\right)\, \rd s\ . 
\end{equation}
Furthermore, the left hand side in Equation~\eqref{eqn:fit_is} is exponentially small if, and only if, both terms in the right hand side of \eqref{eqn:fit_is} are exponentially small.  In addition, if $\kappa_N^{-1}$ has a limit when $N\to\infty$, then the approximation can be made uniform:
\begin{equation}
\label{eqn:fit_is_un}
\Phi_N(x)=\phi_N(x)\left[1+\bo{\kappa^{-1}_N\epsilon_N,\kappa_N\xi_N^2,N^{-1}}\right],\quad x\in [\sfrac{1}{N},1]_N.
\end{equation}
Finally, let $\bx\in [\sfrac{1}{N},1]_N$ be the smallest frequency such that $\phi_N(x)\geq\sfrac{1}{N}$. Then, provided that either $\F$ is an interior potential, or that $\F$ is a boundary potential, and $\kappa_N^{-1}=\bo{N^\alpha}$, with $\alpha<\sfrac{1}{2}$, we have the uniform approximation 
\begin{equation}
\label{eqn:fit_is_nb}
\Phi_N(x)=\phi_N(x)\left[1+\bo{\kappa^{-1}_N\epsilon_N,\kappa_N\xi_N^2,\kappa_N^{1-b}\xi_N}\right],\quad x\in [\bx,1]_N.
\end{equation}
\end{theorem}

The proof of this result is given in Appendix~\ref{ap:deriv:fit_is}.

\begin{remark}
We point out that the approximation \eqref{eqn:asymp_frm} agrees with the exact solution on the frequencies for which the fixation probability is negligible. Finally, if we have an interior potential or if $\kappa_N^{-1}$ does not grow too fast in the case of a boundary potential, then \eqref{eqn:asymp_frm} can correctly identify fixation probabilities that are close to neutral.

Observe also that the bound in \eqref{eqn:fit_is_un} is a relative one rather than an absolute one. Namely,  we can rewrite \eqref{eqn:fit_is_un}  as
\[
\left|\frac{\Phi_N(x)-\phi_N(x)}{\phi_N(x)}\right|=\bo{\kappa^{-1}_N\epsilon_N,\frac{\kappa^{-1}_N}{N^k},\frac{1}{N}}.
\]
 Hence, the approximation is uniformly accurate, in the number of correct digits,  for the entire range of fixation probabilities. This will be of importance in Section~\ref{sec:ess_lp}, where we shall use \eqref{eqn:asymp_frm} to compare possibly very small fixation probability values. 
\end{remark}

	\begin{remark}
The additional  assumption that $\theta$ does not vanish at the boundaries is not essential, but it simplifies the presentation. For instance, if $\F$ is a boundary potential, if $\theta$ vanishes at the boundary (or boundaries) where the maxima of $\F$ is attained, and if $\theta'$ does not vanish at these boundaries, then we can adapt the  calculations for the case of an interior potential, and obtain similar asymptotics.
	\end{remark}

\subsection{Examples}

We now give some examples  

\begin{description}
\item[2 person games with weak-selection] Assume a pay-off matrix specified by 
\begin{center}
\begin{tabular}{l|cc}
&$\A$&$\B$\\\hline
$\A$&$a$&$b$\\
$\B$&$c$&$d$
\end{tabular}
\end{center}
with $a,b,c,d>0$ and fitnesses functions given by $\Psi^{(\A,\B)}=1+\omega(N)\varphi^{(\A,\B)}$, with 
\[
\varphi^{(\A)}(x)=ax+b(1-x)\quad\text{and}\quad\varphi^{(\B)}(x)=cx+d(1-x).
\]
Hence
\begin{align*}
\Theta_N(x)&=\omega(N)\left(\varphi^{(\A)}(x)-\varphi^{(\B)}(x)\right)+\bo{\omega(N)^2},\\
\|\Theta_N\|_\infty&=\omega(N)\|\varphi^{(\A)}-\varphi^{(\B)}\|_\infty+\bo{\omega(N)^2}.
\end{align*}
In this case,
\[
\theta(x)=\gamma((a-b-c+d)x+b-d),\quad \F(x)=-\gamma\left((a-b-c+d)\frac{x^2}{2}+(b-d)x\right),
\]
where $\gamma>0$ is a normalising constant so that $\|\theta\|_\infty=1$.
Recall that if $b-d>0$ ($<0$) and $a-c>0$ ($<0$, respect.) then $\theta>0$ ($<0$, respect.) throughout $[0,1]$, and we have what is usually termed dominance by $\A$ (by \B, respect.). If $b-d<0$ and $a-c>0$, then there is an interior equilibrium $\xs=\sfrac{(d-b)}{(a-b-c+d)}$, with $\theta'>0$, and this is known as a  coordination case. On the other hand, if $b-d>0$ and $a-c<0$, then $\xs$ as given above is again an equilibrium, and in this case $\theta'<0$; this case is known as a coexistence case. 
\item[2 person games, direct identification and no weak-selection] 
In this case, we have
\[
\Theta_N(x)=\Theta(x)=\log\left[\frac{ax+b(1-x)}{cx+d(1-x)}\right]. 
\]
In this case, we have $\theta=\gamma\Theta(x)$, where $\gamma$ is again a normalising constant so that $\|\theta\|_\infty=1$.
Notice that we will have moderate selection, if
\[
\left|\frac{a}{c}-1\right|\ll 1 \quad\text{and}\quad \left|\frac{c}{d}-1\right|\ll 1,
\]
which is attainable, for instance, if $a,b,c,d$ are large---but notice that there is no requirement of grow with $N$.
\item[$d$-player games, with weak-selection] 
In~\citet{Gokhale:Traulsen:2010}, two person games are extended to $d$ person games for Moran process and~\citet{Lessard:2011} further extended they to exchangeable process in the domain of Kingman's coalescence.
Pay-offs are given by
\[
\pi^{(\A)}(j)=\sum_{k=0}^{d-1}\frac{\binom{j-1}{k}\binom{N-j}{d-1-k}}{\binom{N-1}{d-1}}a_k,\qquad 
\pi^{(\B)}(j)=\sum_{k=0}^{d-1}\frac{\binom{j-1}{k}\binom{N-j}{d-1-k}}{\binom{N-1}{d-1}}b_k,
\]
where, as before, $j$ is the number of type $\A$ individuals in a population of size $N$. Furthermore, $a_k$ ($b_k$) is the payoff of a player of type $\A$ ($\B$, respect.) that is playing a game with a group with $k$ players of type $\A$ and $d-k-1$ players of type $\B$.
This formulation contains the first example, if we set $d=2$, and identify $a_1=a$, $a_0=b$, $b_1=c$ and $b_0=d$.

For large $N$, pay-offs are approximated by
%
\begin{align*}
\varphi^{(\A)}(x)&=\sum_{k=0}^{d-1}\binom{d-1}{k}x^k(1-x)^{d-1-k}a_k\ ,\\
\varphi^{(\B)}(x)&=\sum_{k=0}^{d-1}\binom{d-1}{k}x^k(1-x)^{d-1-k}b_k.
\end{align*}
Therefore,
\begin{equation}\label{eq:dgames}
\theta(x)=\gamma\left(\varphi^{(\A)}(x)-\varphi^{(\B)}(x)\right)=\gamma\sum_{k=0}^{d-1}\binom{d-1}{k}x^k(1-x)^{d-1-k}(a_k-b_k),
\end{equation}
where, again, $\gamma$ is a proper normalisation constant.
\end{description}

\subsection{The different evolutionary regimes}

The nature of the approximation \eqref{eqn:fit_is}, and consequently the nature of the fixation patterns, will depend on the behaviour of $\kappa_N$, as $N\to\infty$. In order to make analytical progress, we  will assume that $\kappa_N^{-1}$  has a limit---where we will abuse language and allow it to be infinity---when $N\to\infty$.
Let 
\[
 \kappa^{-1}_\infty \bydef\lim_{N\to\infty}\kappa_N^{-1}.
\]
If $\kappa^{-1}_\infty=\infty$, then for large $N$ we have $\kappa_N^{-1}$ also very large. Analogously, if $\kappa^{-1}_\infty=0$, we have $\kappa_N^{-1}$ very small for large $N$. In the former case, we shall say that evolution is in the  \textbf{selection-driven} regime, while in the latter we shall say that it is in the \textbf{quasi-neutral} regime. Notice that we can have a selection-driven regime, within weak-selection. A possible  classification scheme is given in Table~\ref{tab:evol_reg}.

\begin{table}[htbp]
\begin{tabular}{l||p{2cm}|p{2cm}|p{3.0cm}|p{4cm}}
$\kappa^{-1}_\infty$&Infinite population&Large finite population&Infinite population dynamics&Notes\\\hline\hline
$\infty$&Deterministic&Selection-driven&for certain scales with weak-selection: the replicator dynamics &Possible with or without weak-selection. For a derivation of the infinite population limit with appropriate  scalings see \citet{ChalubSouza09b,ChalubSouza_JMB} .\\\hline
$\ord(1)$&Balanced&Balanced&Replicator-diffusion&Need weak-selection. See \citet{ChalubSouza09b,ChalubSouza_JMB,ChampagnatFerriereMeleard_TPB2006}.\\\hline
0&Neutral&Quasi-neutral&Pure diffusion&See \citet{ChalubSouza09b,ChalubSouza_JMB}.
\end{tabular}
\caption{Different regimes for both finite and infinite populations. }
\label{tab:evol_reg}
\end{table}

In the cases that we dot not have $\kappa_N=\ord(1)$ for large $N$, we shall say that the corresponding regime is not balanced. It turns out that we can describe the non-balanced regimes in a very complete way, and we proceed to do so as follows.

\section{Selection-driven  regime asymptotics}
\label{sec:sdr_asymp}

We shall now assume that we have a selection driven evolutionary dynamics, and choose $N$ to be fixed and large so that $\kappa_N^{-1}\gg1$. For such a fixed $N$, we will write  $\kappa=\kappa_N$ and, from now on, we will  indicate the dependence on $\kappa$ rather than on $N$. 

Throughout this Section we will assume that $\theta$  has at most one zero in $[0,1]$. Furthermore, if we assume that $\theta$ is monotonic, then all error terms can be taken be exponentially small.
Derivations of the asymptotic expressions presented in this Section are provided in Appendix~\ref{ap:asymp_proof}

\subsection{Dominance}
For dominance of one of the types, we have either that $\theta(x)>0$,  $\A$ is dominant, or $\theta(x)<0$, $\B$ is dominant. In the former case, we obtain:
\begin{equation}
\phi_\kappa(x)=1-\exp(-\theta(0)x/\kappa)+\bo{\kappa}.
\label{eqn:dombya}
\end{equation}
In the latter case , we find:
\begin{equation}
 \phi_{\kappa}(x)=\exp(\theta(1)(1-x)/\kappa)+\bo{\kappa}.
\label{eqn:dombyb}
\end{equation}
The dominance case has been studied in a number of regimes by~\citet{Kimura}, and also  by \citet{Gillespie:1981}. The approximations presented in the following recover the previous results when in the appropriate regimes, but also inherit the larger validity of \eqref{eqn:fit_is}.
\begin{remark}
In the context of the derivation of \eqref{eqn:fit_is}, $\kappa^{-1}_{N}$ provides an alternative definition of effective population size. In the nonneutral case, effective population size is not equal to the population size, and this has some implication in the context of invasions. As an example, we consider dominance by $\A$. In this case, from  equation~\eqref{eqn:dombya}, we then have
\[
\phi_\kappa(1/N)=1-\exp\left(-\|\Theta_N\|_\infty\theta(0)\right)+\bo{\kappa}.
\]
In general, $\|\Theta_N\|_\infty$ can have a strong dependence on $N$. Thus, it can  be either very large  or small and not just order one as it  would be  expected from the non-frequency dependent case.
\end{remark}

\subsection{Coexistence}

In this case, $\theta$ has a unique zero $\xs$, $\theta'(\xs)<0$. It turns out that  the crucial quantity  for understanding the fixation in this case is the value of the fitness potential when \A is fixed in the population, $\F(1)$. Indeed, we shall have three possible cases as follows: 

\begin{description}
 \item[$\F(1)\ll-\kappa$] the asymptotic approximation is given by \eqref{eqn:dombya};

\item[$\F(1)\gg\kappa$] the asymptotic approximation is given by \eqref{eqn:dombyb};

\item[$|\F(1)| \sim \kappa$] In this case, let
\[
C=\exp(\F(1)/\kappa)
\text{, and }
\gamma=\frac{|\theta(1)|}{\theta(0)}
\]
 Then the asymptotic approximation is given by
\begin{equation}
\phi_\kappa(x)=\frac{C}{C+\gamma}\exp(\theta(1)(1-x)/\kappa)+\frac{\gamma}{C+\gamma}\left(1-\exp(-\theta(0)x/\kappa)\right) +\bo{\kappa},
 \label{eqn:fp_asymp_cxs}
\end{equation}
with $\theta(0)>0>\theta(1)$. 
\end{description}
\begin{remark}
\label{rmk:cc}
Equation~\eqref{eqn:fp_asymp_cxs} can be seen as a convex combination of the equations \eqref{eqn:dombya} and \eqref{eqn:dombyb}, and hence in a sense the coexistence case interpolates between the two dominance cases. Furthermore, we shall see in Theorem~\ref{thm:rd} that $C=\sfrac{\rho_\B}{\rho_\A}$, where $\rho_{\A,\B}$ are the invasion probabilities of types \A  (\B).
\end{remark}

\begin{remark}
Note that if $\F(1)\to-\infty$ ($\F(1)\to\infty$), then $C\to 0$ ($C\to\infty$, respectively) and therefore equation (\ref{eqn:fp_asymp_cxs}) reduces to \eqref{eqn:dombya} (to \eqref{eqn:dombyb}, respectively).
Notice also that the graph of \eqref{eqn:fp_asymp_cxs}, for small $\kappa$, and apart from boundary layers of order $\kappa$ at the endpoints, is  essentially a horizontal line at the level $\gamma/(C+\gamma)$. We shall denote this level by the Fixation Probability Plateau (FPP). 
\end{remark}

\subsection{Coordination}

In the coordination case, $\theta$ also has a unique root $x^*$, with $\theta'(\xs)>0$. We have:
\begin{equation}
\phi_{\kappa}(x)=\frac{\N\left(\sqrt{\frac{\theta'(\xs)}{\kappa}}(x-\xs)\right)-\N\left(-\sqrt{\frac{\theta'(\xs)}{\kappa}}\xs\right)}{\N\left(\sqrt{\frac{\theta'(\xs)}{\kappa}}(1-\xs)\right)-\N\left(-\sqrt{\frac{\theta'(\xs)}{\kappa}}\xs\right)}+\bo{\sqrt{\kappa}},
 \label{eqn:fp_asymp:cd}
\end{equation}
where $\N(x)=\frac{1}{\sqrt{2\pi}}\int_{-\infty}^x\e^{-y^2/2}\rd y$ is the normal cumulative distribution.
If, in addition, we have that $\xs\gg\sqrt{\kappa}$, and $1-\xs\gg\sqrt{k}$ then \eqref{eqn:fp_asymp:cd} can be simplified to
\begin{equation}
 \phi_{\kappa}(x)=\N\left(\sqrt{\frac{\theta'(\xs)}{\kappa}}(x-\xs)\right)+\bo{\sqrt{\kappa}}.
\end{equation}
Thus, for $\xs$ far from the endpoints we have the interesting result that
\[
 \phi_{\kappa}(\xs)=\frac{1}{2}+\bo{\sqrt{\kappa}}.
\]
On the other hand, if $\xs=\sqrt{\kappa}\bx$ with $\bx=\bo{1}$, we then have that
\[
 \phi_{\kappa}(\xs)=\frac{\frac{1}{2}-\N\left(-\sqrt{\theta'(\xs)}\bx\right)}{1-\N\left(-\sqrt{\theta'(\xs)}\bx\right)}+\bo{\sqrt{\kappa}}.
\]
where as, if $\xs=\kappa\bx$ instead, we find that
\[
 \phi_{\kappa}(\xs)=\frac{\bx\sqrt{\sfrac{2\kappa\theta'(\xs)}{\pi}}}{1+ \bx\sqrt{\sfrac{2\kappa\theta'(\xs)}{\pi}}} + \bo{\kappa}.
\]

In this formula, we used that, in this case,  the correction term is of order $\kappa$.

\subsection{Remarks on the formulae}

Let $f(x,t)$ be the probability that type \A reaches fixation at time $t$ or before given its initial presence $x$. Then, the so called Kimura equation~\citep{Kimura} is given by
\begin{equation}\label{eq:Kimura}
 \partial_t f=kx(1-x)\partial_x^2f+cx(1-x)\partial_xf\ ,
 \end{equation}
with $f(0,t)=0$ and $f(1,t)=1$. Formulas \eqref{eqn:dombya} and \eqref{eqn:dombyb} are stationary  solutions of equation~(\ref{eq:Kimura}) for constant fitness differences, i.e., $\Theta_N(x)=c$, with $c>0$ for dominance by $\A$, and $c<0$ for dominance by $\B$. Equation \eqref{eqn:fp_asymp:cd} is the exact solution for the case of linear fitness difference, and weak-selection. In this sense, these results show that an arbitrary pay-off difference, with at most one root in the unit interval, is equivalent to a linear fitness difference with the same signal pattern across the unit interval. The case of coexistence is different altogether. Nevertheless, we would like to point out a kind of duality between coordination and coexistence regarding $\phi_{\kappa}(x)$ for different $x$, and $\phi_{\kappa}(\xs)$ for different $\xs$ as shown in Figure~\ref{fig:dual}.

\begin{figure}[htbp]
\begin{center}
\subfloat[Fixation for an observed frequency  $x$. \label{fig_dual_a}]{\includegraphics[scale=0.35]{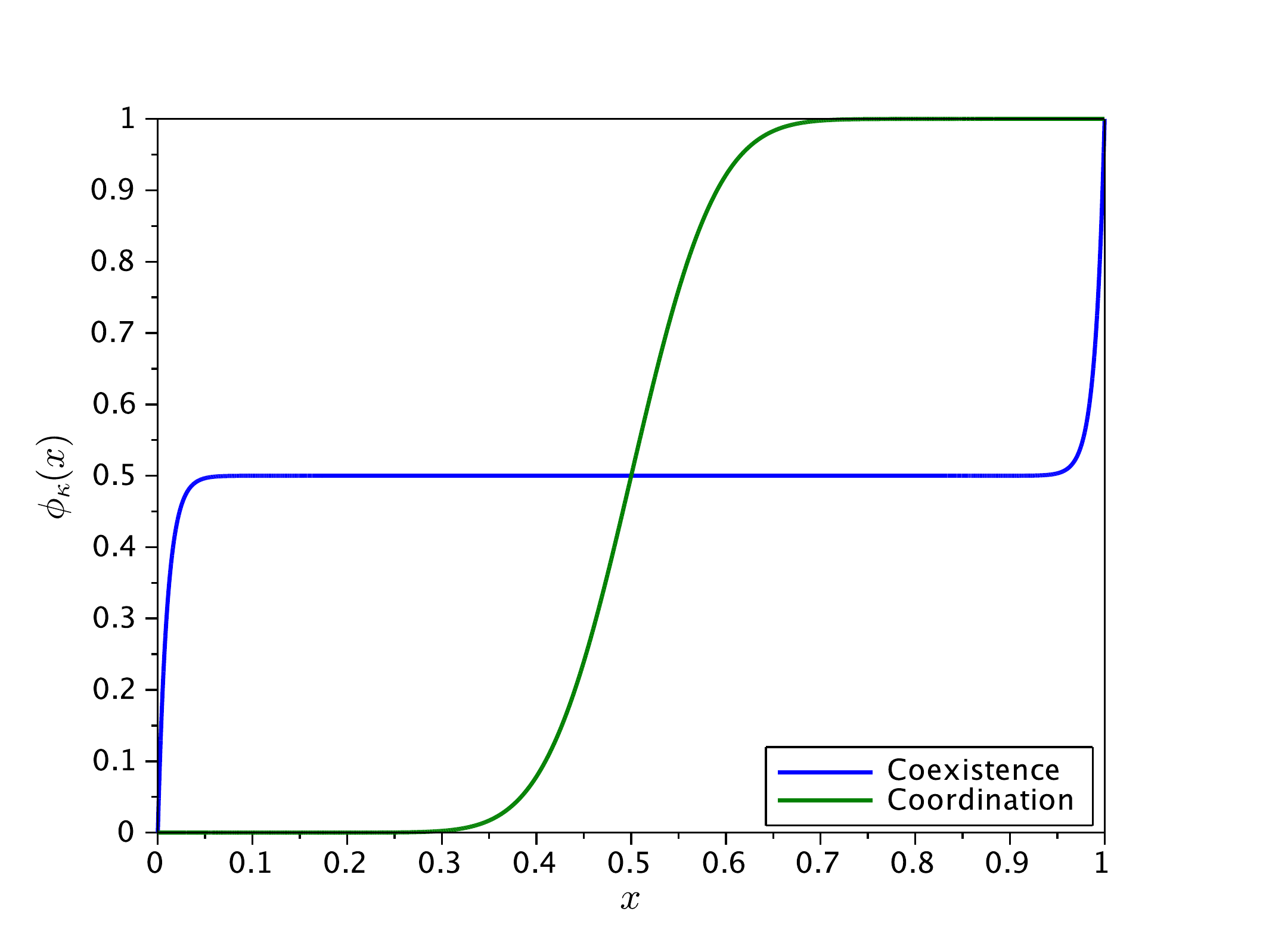}}
\subfloat[Fixation at the equilibrium point, $\xs$, as this equilibrium runs throughout the unit interval.  \label{fig:dual_b}]{\includegraphics[scale=0.35]{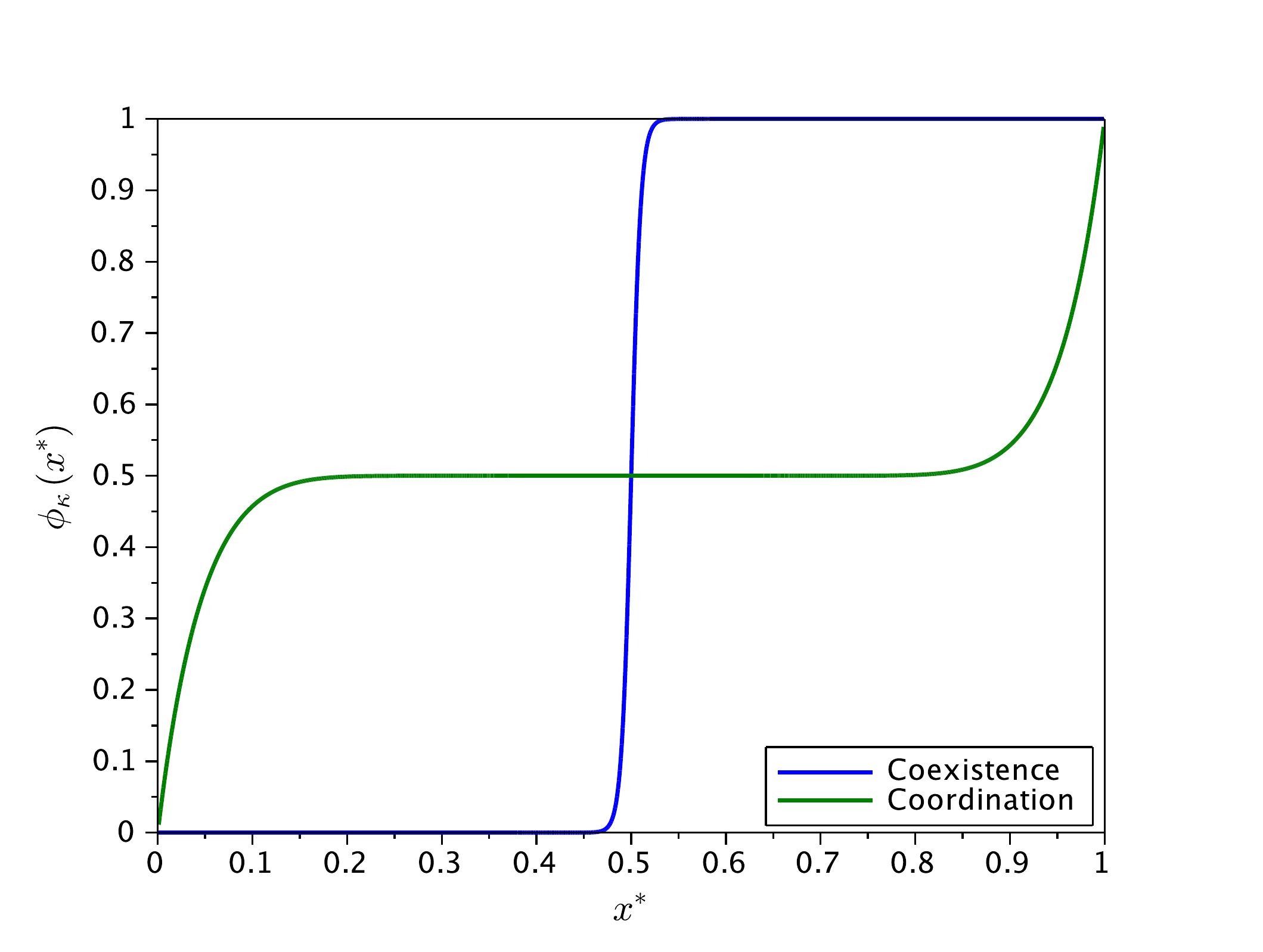}}
\end{center}
\caption{Duality between coexistence and coordination fixation patterns in the linear fitness case. The coexistence fixation when $\xs=\sfrac{1}{2}$ is similar to fixation at the equilibrium point in the case of coordination, as this equilibrium runs the unit interval. Analogously, the coordination fixation, also when $\xs=\sfrac{1}{2}$, is similar to the fixation at equilibrium for the coexistence case, for different values of $\xs$ in [0,1]. \label{fig:dual}}
\end{figure}

\section{Fixation in the presence of a deterministic interior ESS}

\label{sec:fde}

We say that two types \emph{coexist} if there is a stable state where both types are present. In the deterministic approach --- namely, in the replicator dynamics case~\citep{HofbauerSigmund1998} --- this happens only if there is a point $x^*\in(0,1)$ such that $\Psi^{(\A)}(x^*)=\Psi^{(\B)}(x^*)$. However, this condition says nothing about the stability of this point under small perturbations. We say that a point $x^*$ is an asymptotically stable equilibrium, if it is an equilibrium and sufficiently close states will be attracted to $x^*$. It can be shown  that an equilibrium $x^*\in(0,1)$ is asymptotically stable if 
there exists $\eps_0>0$, such that $\pm\mathop{\mathrm{sgn}}\left[(\Psi^{(\A)}(x^{*}\pm\eps)-\Psi^{(\B)}(x^{*}\pm\eps))\right]<0$ for $0<\eps<\eps_0$~\citep{HofbauerSigmund1998}.
This is the origin of the concept of \emph{evolutionary stable strategy}, or ESS~\citep{MaynardSmith1988}.

However, in the finite population case, there is no non-trivial stable state~\citep{Karlin_Taylor_first}. Therefore, the coexistence case is where both modelling paradigms --- finite population (stochastic) and infinite population (deterministic) --- differ more markedly. The solution to this apparent contradiction is the fact that, in this case, the finite population model has a quasi-stationary distribution that leads to the existence of a metastable dynamics which lasts for a time that increases exponentially with the population size~\citep{ChalubSouza_JMB,AntalScheuring_2006,MeleardVillemonais2011}. In the one-dimensional setting with a coexistence equilibrium that is far from the boundaries, this dynamics can be approximately described by an Ornstein-Ullenbeck process with a long term mean around the coexistence equilibrium \citep{VanKampen}.

Intuitively, one might expect, since the dynamics is most likely to develop in the vicinity of the deterministic equilibrium, that the fixation is largely independent of the state that one is observing in a certain moment, and more related to the equilibrium level.

As we shall now see, while the former intuition is correct, the latter is not.

\subsection{The near one-half law}

We shall now want to study the dynamics in the case of coexistence, and weak-selection, but selection-driven regime. We consider the following fitness differences:
\[
\Theta_N(x)=\log\frac{1+N^{-\alpha}\left[ax+b(1-x)\right]}{1+N^{-\alpha}\left[cx+d(1-x)\right]},
\]
with $\alpha>0$ and $a,b,c,d>0$ associated to a  pay-off matrix $\left(\begin{smallmatrix}a&b\\c&d\end{smallmatrix}\right)$ that yields a coexistence equilibrium, i.e., $c>a$, $b>d$. This is equivalent to specify
\[
\Theta_N(x)=N^{-\alpha}\gamma(\xs-x)+\bo{N^{-2\alpha}} ,
\]
with  $\gamma>0$ and $\xs\in (0,1)$.
We have the following result:
\begin{theorem}
\label{thm:nohl}
Assume the we are in the coexistence case, weak-selection but selection-driven regime. Assume also that we have linear formal limit fitness differences, i.e.,
\[
\theta(x)=\bar\gamma(\xs-x),\quad \xs\in(0,1).
\]
where $\bar\gamma\bydef\frac{1}{\max\{x^*,1-x^*\}}$.

Then, we have the following scenarios, with $c$ being a constant $\bo{1}$:
There are values $0<x_1<y_1<\sfrac{1}{2}<y_2<x_2<1$, such that $x_1$ is near zero, $x_2$ is near one,  both $y_1$ and $y_2$ are near $\sfrac{1}{2}$ and  with the property that if
\begin{description}
\item[$\xs<y_1$] Then, for all $x<x_2$, the fixation probability of \B is near unity.
\item[$\xs=\sfrac{1}{2}$] Then, for all $x\in(x_1,x_2)$, we have near $\sfrac{1}{2}$ probability of fixation for both types.
\item[$\xs>y_2$] Then, for all $x>x_1$, we have that the fixation probability of $\A$ is near unity.
\end{description}
Here, the statement $X$ \emph{is near} $Y$ means that $|X-Y|=\bo{\kappa}$.
\end{theorem}
\begin{proof}
This follows from the direct calculation that
%
$\F(1)=-\bar\gamma(\xs-\sfrac{1}{2})$.
%
The results then follow from equation \eqref{eqn:fp_asymp_cxs}.
\end{proof}

As an illustration of this behaviour, we show in Figure~\ref{fig:fixp_vxs} the fixation probability for the $\Theta_N$ considered in Theorem~\ref{thm:nohl}. The variation with $\xs$ observed in this figure suggest that the fixation pattern changes very fast in the vicinity of the $\xs=\sfrac{1}{2}$ from dominance of $\B$ to dominance of $\A$.  Thus a \textit{coexistence layer}, in the sense that both types have a significant probability of fixation, exists only when $\xs$ is close to one-half. This behaviour is further illustrated  Figure~\ref{fig:coex_layer}, where  we plot $\Phi_{N}(x^*)$ for $x^*\in(0,1)$.

\begin{figure}[htbp]
\begin{center}
\includegraphics[scale=0.55]{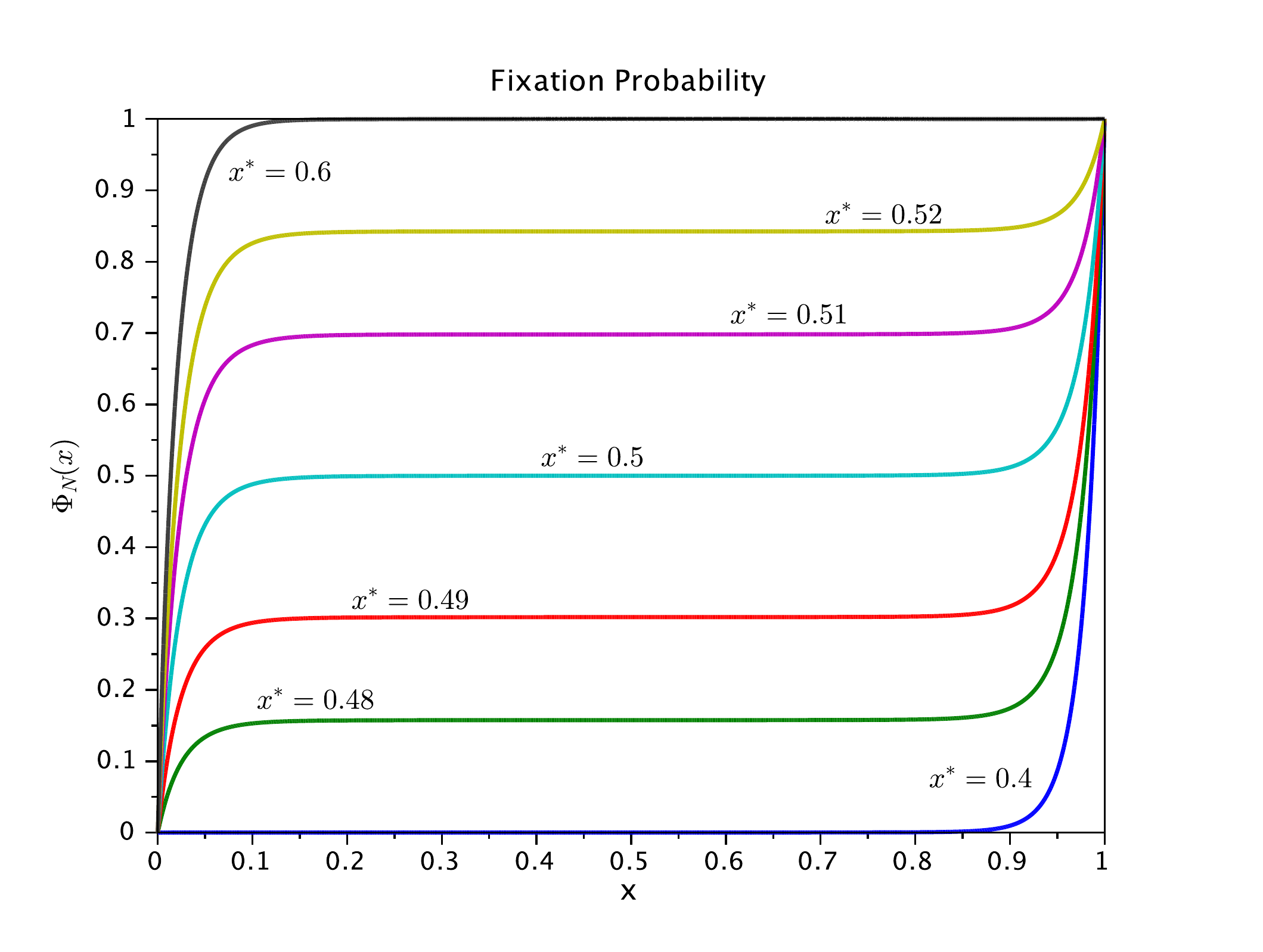}
\caption[]{Fixation probability  computed from~\eqref{eqn:fix_disc} using  payoff matrix given by
{\centering
\begin{tabular}{r|cc}
&\A&\B\\\hline
\A&1&$1+\xs$\\
\B&$2-\xs$&1
\end{tabular}
\par}
with $N=1000$, and $\alpha=\sfrac{1}{3}$. The corresponding generalised log-fitness are given by  $\Theta_N(x)=N^{-1/3}(\xs-x)+\bo{N^{-2/3}}$.}
\label{fig:fixp_vxs}
\end{center}
\end{figure}

\begin{figure}[htbp]
\begin{center}
\includegraphics[scale=0.55]{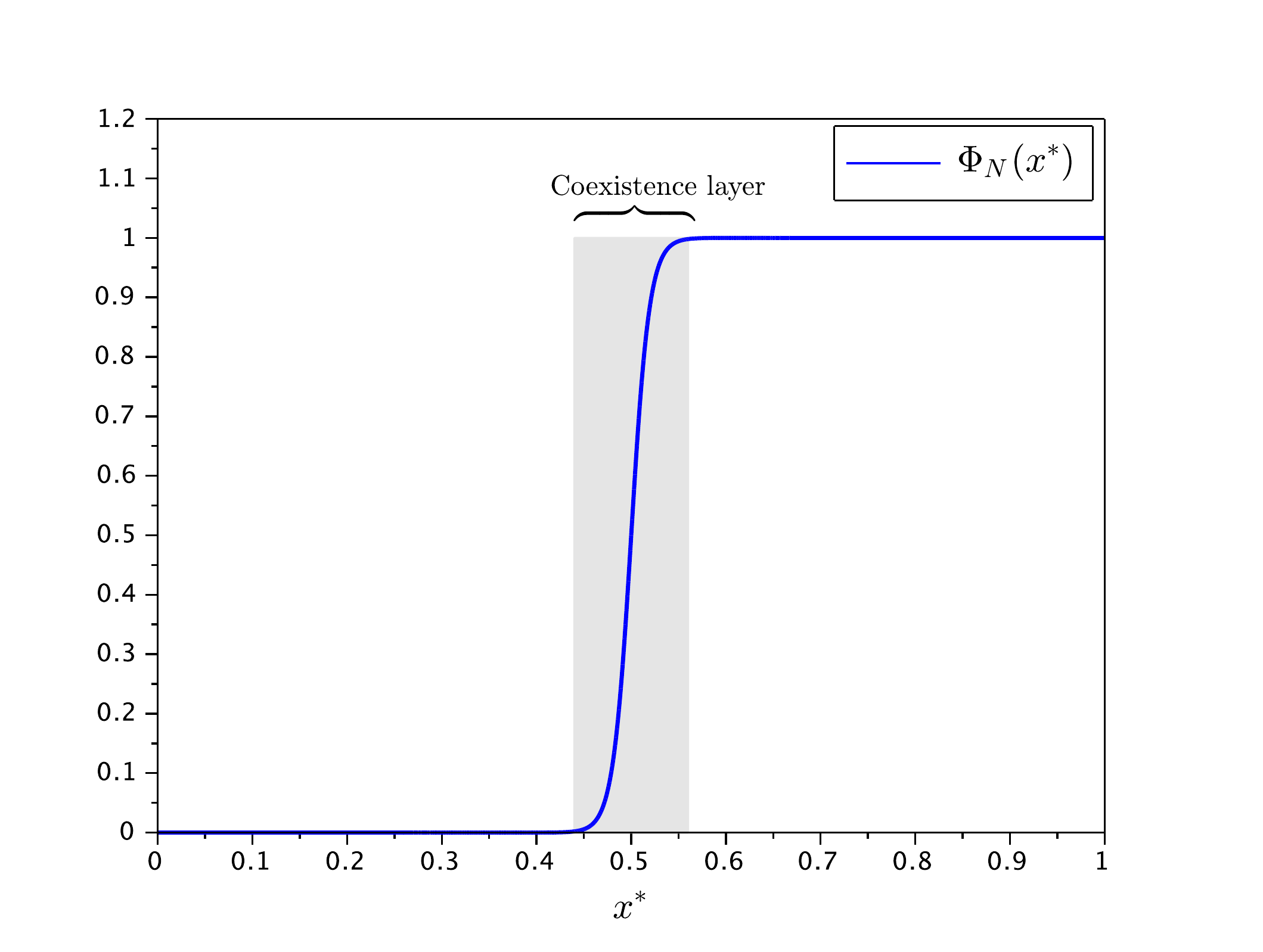}
\caption{Fixation probability at the the equilibrium point computed using \eqref{eqn:fix_disc}, with $N=1000$ and $\Theta_N(x)=N^{-1/3}(\xs-x)+\bo{N^{-2/3}}$. The layer effect is enhanced if either one considers a larger population, or  a slower  convergence towards zero in the fitness difference.}
\label{fig:coex_layer}
\end{center}
\end{figure}

\subsection{Fixation in the absence of weak-selection}

In the absence of weak-selection, the fixation behaviour for the coexistence case can be much more varied. 

Consider the payoff matrix, independent of $N$, given by

\begin{center}
\begin{tabular}{r|cc}
&\A&\B\\\hline
\A&$a$&$b$\\
\B&$c$&$d$
\end{tabular}
\end{center}
with $a,b,c,d>0$.

Assuming self-interaction, i.e., that each individual plays the game against him/herself, then the corresponding log-difference of the fitness is given by
\[
\Theta_{N}(x)=\log\left[\frac{(a-b)x+b}{(c-d)x+d}\right],
\]
which is independent of $N$. We assume, as before, $c>a$, $b>d$. Notice that, in this case, we have 
\[
\theta(x)=\gamma\Theta_{N}(x),\quad \gamma^{-1}=\max(|\log(\sfrac{a}{c})|,|\log(\sfrac{b}{d})|).
\]

From which we then obtain
\[
\F(1)=-\gamma\left\{\frac{a\log a - b\log b}{a-b} -\frac{c\log c -d\log d}{c-d}\right\}.
\]
Notice that the above expression is neither translation or multiplicative invariant. Hence, as far as fixation is concerned, the absolute values of the entries in the payoff matrix --- and not only their relative values --- are important. 

\begin{remark}
We point out that the form of $\Theta_N$ used above is a slight simplification, since the usual assumption is that an individual does not play against himself. In this more conventional case, i.e., without self-interaction, we find
\[
\Theta_{N}(x)=\log\left[\frac{(a-b)x+b-\sfrac{a}{N}}{(c-d)x+d(1-\sfrac{1}{N})}\right].
\]
Notice, however, that the difference between both fitness log-differences  is $\bo{\sfrac{1}{N}}$, and hence the fixation pattern is the same, if $N$ is sufficient large. In particular,  the continuous approximation does not change.
\end{remark}

As an example, consider the particular payoff matrix:

\begin{center}
\begin{tabular}{r|cc}
&\A&\B\\\hline
\A&11&110.075\\
\B&11.025&110
\end{tabular}
\end{center}

Then, $\Theta_{N}(0)=\log(1.00068)>0$ and $\Theta_{N}(1)=\log(0.99773)<0$.  Hence we are in the coexistence case, and  the unique equilibrium is at $\xs=3/4$.

As for the fixation pattern,  a  direct computation yields
$\F(1)=-0.0749870$.
Thus, the fixation of $\A$ is almost certain, if $N$ is sufficiently large.  On the other hand, if we subtract ten from the previous matrix, we obtain
\begin{center}
\begin{tabular}{r|cc}
&\A&\B\\\hline
\A&1&100.075\\
\B&1.025&100
\end{tabular}
\end{center}
and we now obtain a positive potential at $x=1$, namely
$\F(1)=0.0079851$.
In this case, extinction of $\A$ is now almost certain --- again, if $N$ is large enough. In Figure~\ref{fig:F_var}, we can see the variation effect of adding a constant $c$ to all entries of the payoff-matrix.
\begin{center}
\begin{tabular}{r|cc}
&\A&\B\\\hline
\A&1&50.075\\
\B&1.025&50
\end{tabular}
\end{center}
\begin{figure}[htbp]
\includegraphics[scale=0.55]{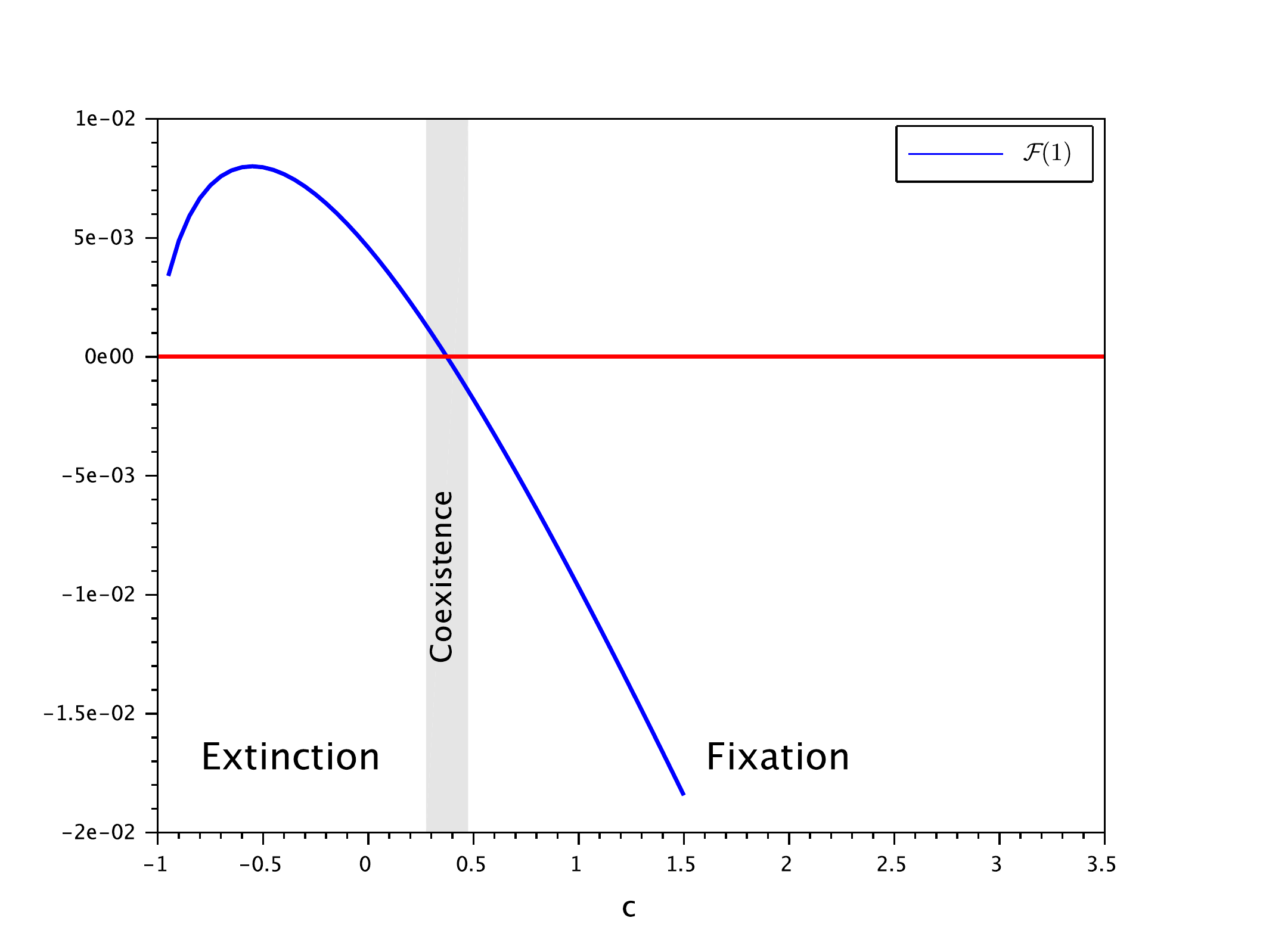}
\caption{Behaviour of $\F(1)$ as a  constant $c$ is added to the  pay-off matrix above. One can then see that, by adding a constant, one can obtain the full spectrum of possible evolutionary dynamics in the case of coexistence dynamics, and that the actual level of equilibrium can have little influence on the outcome. \label{fig:F_var}}
\end{figure}

\subsection{A general result for fixation plateaus}

The previous results also show that\, in the absence of weak-selection, the FPP is structurally unstable for large populations, since a perturbation of order $\kappa$ in in the model parameters might change $\F(1)$ by an order one value, and thus significantly alter the corresponding fixation pattern. We shall now see that this behaviour is by no means exceptional. Notice that this structural instability for large, but finite $N$, while being a finite population effect, is diverse from the difference in fixation patterns usually studies via, for instance, the concept of $\textsc{ESS}_N$; see section~\ref{sec:ess_lp}.

The FPP indicates the likelihood of fixation of either type. If the FPP is close to zero, we have essentially dominance of \B, while if the  FPP is close to unity will have essentially dominance of \A. For values of the FPP  close to $\sfrac{1}{2}$, the evolutionary dynamics is fair in the sense that both types have very similar fixation probabilities. We will now show that the phenomena in Theorem~\ref{thm:nohl} is more robust.

\begin{theorem}\label{thm:oneparameter}
Consider a family of models parametrized by $\lambda\in\Lambda$, an open interval, and let $I=(0,1)$. Let  $\theta:I\times\Lambda\to\R$ be smooth  and write, accordingly, $\theta=\theta(x,\lambda)$ and  $\F=\F(s,\lambda)$. Assume that  $\partial_x\theta(x,\lambda)<0$ in $I\times\Lambda$. Assume further that  there exist  $x_0^*$ and $\lambda_0$ such that  $\theta(x_0^*,\lambda_0)=0$,  $\F(1,\lambda_0)=0$ and with
\[
\I=\int_0^1\partial_\lambda\theta(r,\lambda_0)\,\rd r
\]
being non-zero and order one. Then, there exists an open interval $L=(\lambda_-,\lambda_+)$, and a unique smooth function $x:L\subset\Lambda\to I$ such that $x(\lambda_0)=x_0^*$  and $\theta(x(\lambda),\lambda)=0$. Moreover, in the selection-driven regime, provided that we have $\lambda_0-\lambda_-,\lambda_+-\lambda_0\gg\kappa$, there will be values around $\lambda_0$ such that the $FPP$ will change from near dominance of \A to near dominance of \B.
\end{theorem}
\begin{proof}
Existence of the function of $x(\lambda)$ follows from a standard argument using the implicit function theorem.  Suppose that such a function is not unique. Then there exists $\lambda\in L$ and $x_1(\lambda)<x_2(\lambda)$ such that $\theta(x_1(\lambda),\lambda)=\theta(x_2(\lambda),\lambda)=0$. Hence, there $\partial_x\theta(x,\lambda)$ must vanish in the interval $(x_1,x_2)$ contradicting its negative-definiteness. Writing
\[
\F(1,\lambda)=-(\lambda-\lambda_0)\I+\bo{(\lambda-\lambda_0)^2},
\]
we see, by taking $\lambda$ sufficiently close to $\lambda_-$ or to $\lambda_+$, that  
we can have both $\F(1,\lambda)\ll-\kappa$ and $\F(1,\lambda)\gg\kappa$, for different values of $\lambda$. Hence,  the FPP can have all values from very near zero to very near unity.
\end{proof}

\begin{remark}
Many different forms of Theorem~\ref{thm:oneparameter} can be formulated. The crucial points, for the result to hold, are that $\I=\bo{1}$, and that the interval $(\lambda_-,\lambda_+)$ is sufficiently large in both directions.
\end{remark}

\clearpage

\subsection{Blockage and tunnelling in evolution}

We now take a small detour. For given fitnesses functions, assume the associated deterministic model given by the replicator equations; consider additionally that there are multiple non-trivial equilibria in the replicator equation. For simplicity, let us assume that there are only two equilibria, i.e., $x_1$ and $x_2$ are the only solutions in $(0,1)$ such that $\Psi^{(\A)}(x)=\Psi^{(\B)}(x)$. As before in this section, $0<\kappa\ll 1$. We assume, in case I, that $x_1$ is unstable and $x_2$ is stable and in case II, the other way round.

\subsubsection{Case I} For case I, we have that $x_1$ is a local maximum for the fitness potential, and $x_2$ is a local minimum. Therefore, the global maximum of $\F$ can be either at $x=x_1$ or at $x=1$. Hence, we have the following subcases:
\begin{description}
\item $\F(1)-\F(x_1)\ll-\kappa$  Fixation is given  by \eqref{eqn:fp_asymp:cd}, with $\xs$ replaced by $x_1$.;
\item $\F(1)-\F(x_1)\gg\kappa$ Fixation is given by \eqref{eqn:dombyb};
\item $\F(1)-\F(x_1)\sim\kappa$ Let
\[
C=\exp\left(\frac{\F(1)-\F(x_1)}{\kappa}\right)
\quad\text{and}\quad
 \gamma =\frac{|\theta(1)|}{2\pi}\left[\N(\sigma^{-1}(1-x_1))-\N(-\sigma^{-1}x_1)\right],
\]
with $\sigma^{-1}=\sqrt{\sfrac{\theta'(x_1)}{\kappa}}$.
Then, we have:
\begin{equation}
\label{eqn:fix_two_ci}
\phi_\kappa(x)=\frac{\gamma}{\gamma + C}\frac{\N(\sigma^{-1}(x-x_1))-\N(-\sigma^{-1}x_1)}{\N(\sigma^{-1}(1-x_1))-\N(\sigma^{-1}x_1)}+\frac{C}{\gamma+C}\exp\left(\theta(1)\frac{1-x}{\kappa}\right).
\end{equation}
Notice that, if $C\to\infty$, then \eqref{eqn:fix_two_ci} reduces to \eqref{eqn:dombyb}, whereas if $C\to0$ it reduces to \eqref{eqn:fp_asymp:cd}.
\end{description}

The above result shows that, if we have $\F(1)$ larger than $\F(x_1)$ then the ESS acts as a blockade to the evolution, and hence the extinction of $\A$ is almost certain.

\begin{figure}[htbp]
\subfloat[$\F(1)-\F(x_1) \gg \kappa$]{
\parbox{0.33\linewidth}{\includegraphics[scale=1]{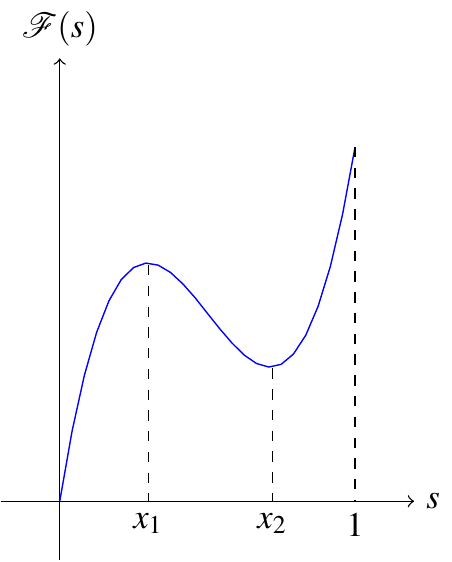}}}
\subfloat[$\F(1)-\F(x_1) \ll -\kappa$]{
\parbox{0.33\linewidth}{\includegraphics[scale=1]{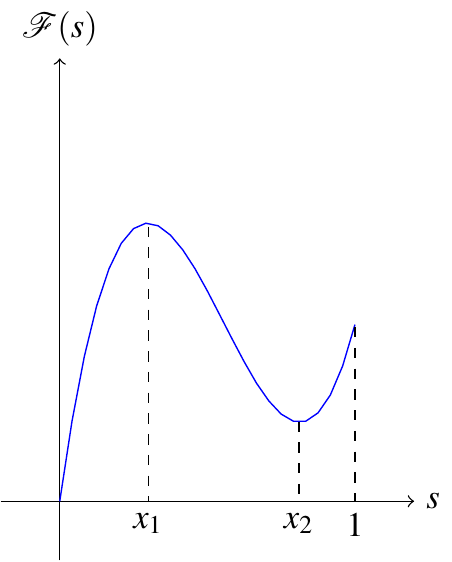}}}
\subfloat[$\F(1)-\F(x_1)\sim\kappa$]{
\parbox{0.33\linewidth}{\includegraphics[scale=1]{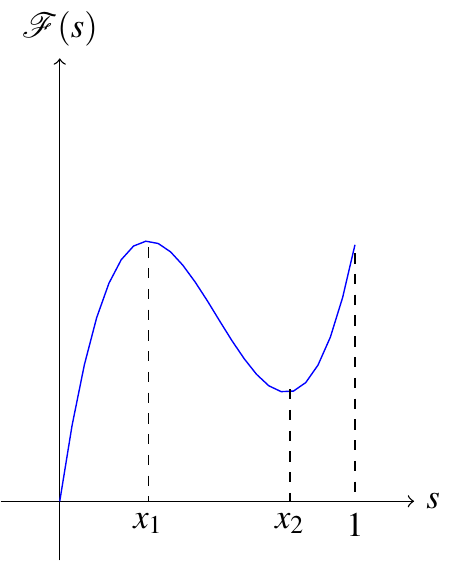}}}
\caption{Relative configuration of maxima and minima for case I. }
\end{figure}

\begin{figure}[htbp]
\begin{center}
\subfloat[Replicator Dynamics flow]{
\parbox{0.45\linewidth}{\includegraphics[scale=1]{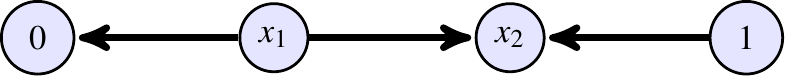}}}	

\vspace{10pt}
\subfloat[Coordination]{
\parbox{0.45\linewidth}{\includegraphics[scale=1]{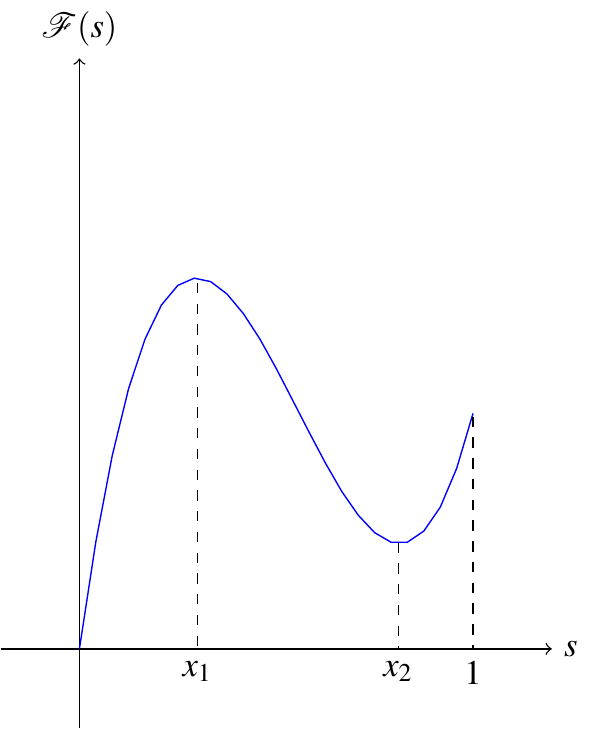}}}
\hfill
\subfloat[Blockage]{
\parbox{0.45\linewidth}{\includegraphics[scale=1]{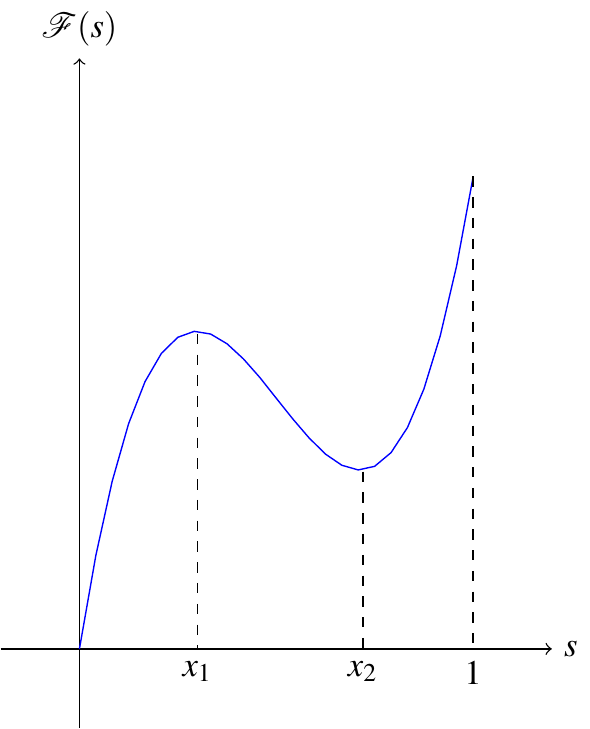}}}

\subfloat[Fixation probability for coordination]{
	\parbox{0.45\linewidth}{\includegraphics[scale=1]{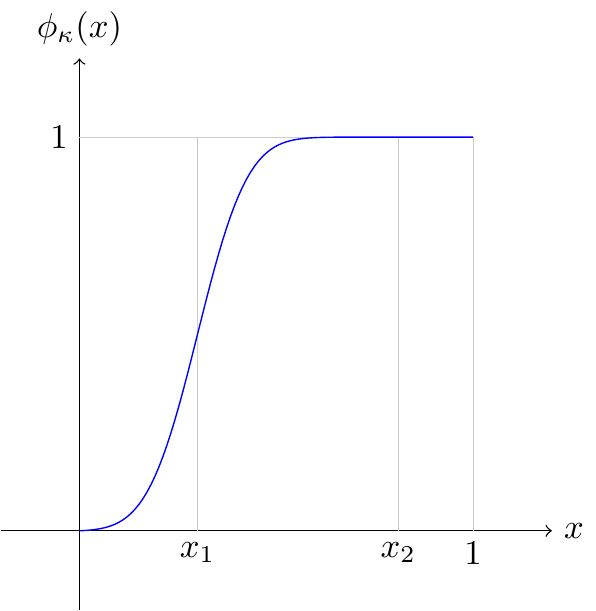}}}
\hfill
\subfloat[Fixation probability for coordination]{
	\parbox{0.45\linewidth}{\includegraphics[scale=1]{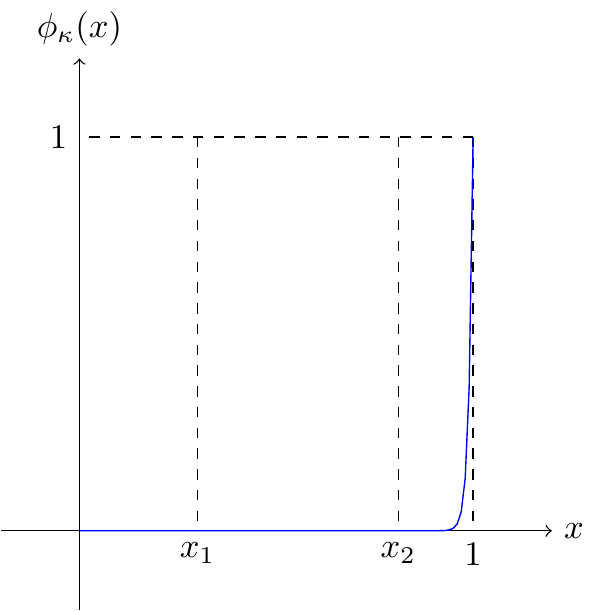}}}
\end{center}
\caption{Infinite versus finite population evolution features. Notice that depending on the landscape of the potential, we can have either a behaviour that is coordination like (figures (\textsc{B}) and (\textsc{D})) or dominance like figures (\textsc{C}) and (\textsc{E})). In the latter case, the ESS will prevent any chance of fixation of \A, and will term this an evolutionary blockage.} 
\end{figure}

\clearpage 

\subsubsection{Case II} For case II,  we have that $x_1$ is a local minimum for the fitness potential, while $x_2$ is a local maximum. Thus the global maximum of $\F$ can occur only at $x=0$ or $x=x_2$. This yields the following subcases:
\begin{description}
\item[$\F(x_2)\gg\kappa$] Fixation is given by \eqref{eqn:fp_asymp:cd}, with $\xs$ replaced by $x_2$.;
\item[$\F(x_2)\ll-\kappa$] Fixation is given by \eqref{eqn:dombya};
\item[$\F(x_2)\sim\kappa$] Let
\[
C=\exp\left(\F(x_2)/\kappa\right)
\quad\text{and}\quad
 \eta  =\frac{\theta(0)}{2\pi}\left[\N(\sigma^{-1}(1-x_2))-\N(-\sigma^{-1}x_2)\right].
\]
Then, we have:
\begin{equation}
\label{eqn:fix_two_cii}
\phi_\kappa(x)=\frac{\eta}{\eta + C}\left(1-\exp\left(-\theta(0)\frac{x}{\kappa}\right)\right)     +\frac{C}{\eta+C}\frac{\N(\sigma^{-1}(x-x_2))-\N(-\sigma^{-1}x_2)}{\N(\sigma^{-1}(1-x_2))-\N(\sigma^{-1}x_2)}.
\end{equation}
Notice that, if $C\to\infty$, then \eqref{eqn:fix_two_cii} reduces to \eqref{eqn:fp_asymp:cd}, whereas if $C\to0$ it reduces to \eqref{eqn:dombya}.
\end{description}

In this case, the ESS equilibria acts as a ``tunnel''  and the dynamics  almost certainly will  cross the evolution barrier imposed by $x_2$, and thus fixation of $\A$ is almost certain.
\begin{figure}[htbp]
\subfloat[$\F(x_2)\gg\kappa$]{
\parbox{0.33\linewidth}{\includegraphics[scale=1]{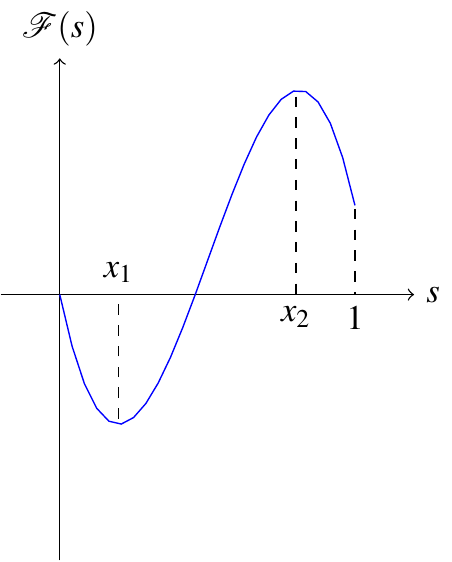}}}
\subfloat[$\F(x_2)\ll-\kappa$]{
\parbox{0.33\linewidth}{\includegraphics[scale=1]{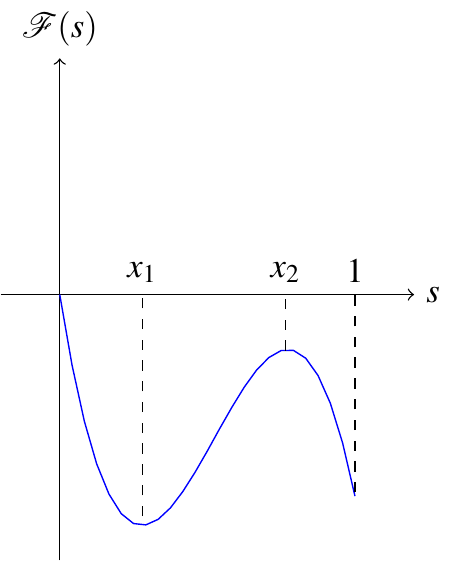}}}
\subfloat[$\F(x_2)\sim\kappa$]{
\parbox{0.33\linewidth}{\includegraphics[scale=1]{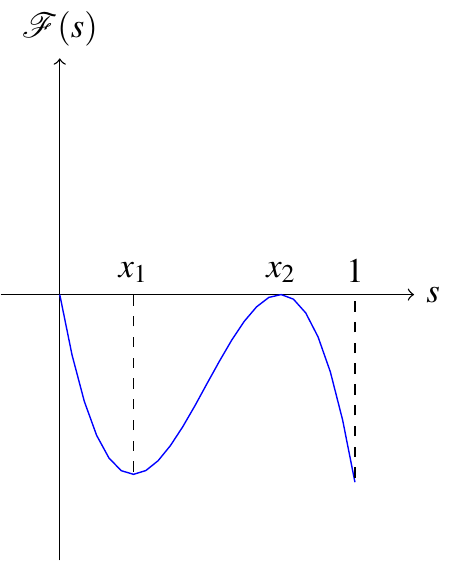}}}

\caption{Relative configuration of maxima and minima for case II. }
\end{figure}

\begin{figure}[htbp]
\begin{center}
\subfloat[Replicator Dynamics flow]{
\parbox{0.45\linewidth}{\includegraphics[scale=1]{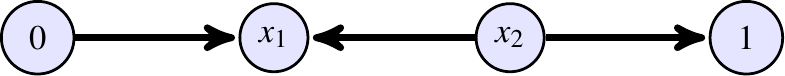}}}

\vspace{10pt}
\subfloat[Coordination]{
\parbox{0.45\linewidth}{\includegraphics[scale=1]{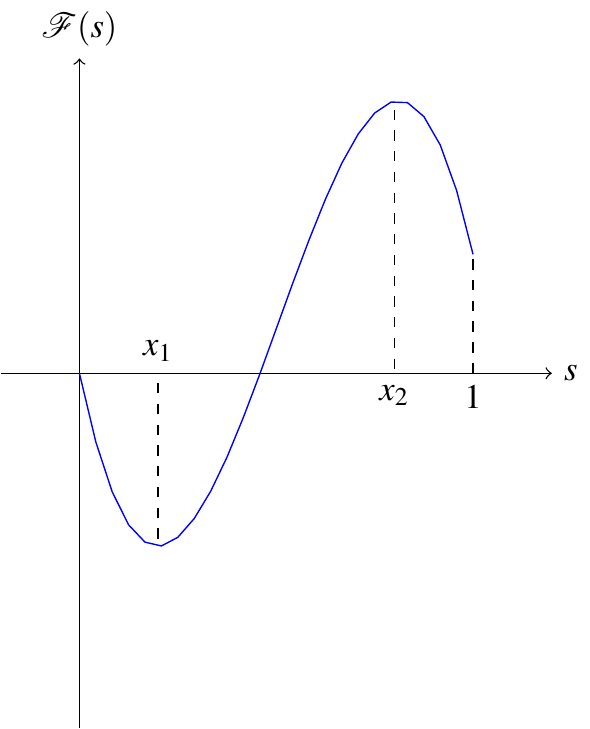}}}
\hfill
\subfloat[Tunnelling]{
\parbox{0.45\linewidth}{\includegraphics[scale=1]{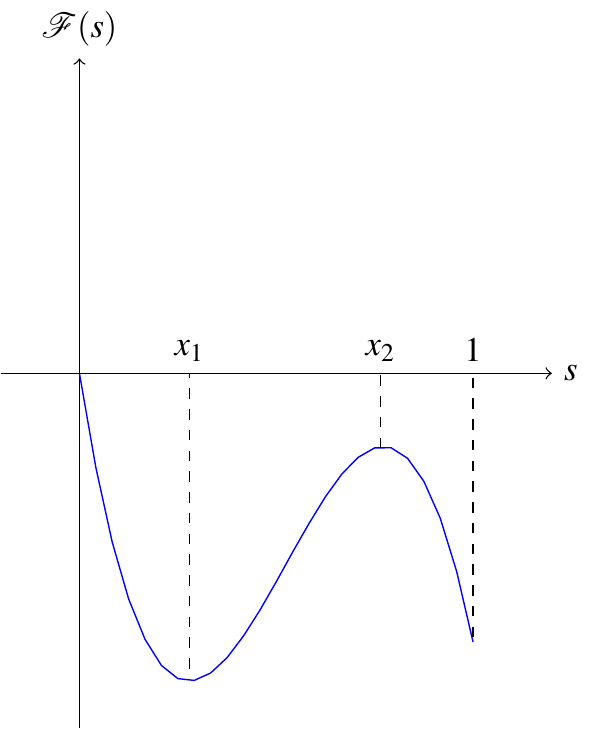}}}

\subfloat[Fixation probability for coordination]{
	\parbox{0.45\linewidth}{\includegraphics[scale=1]{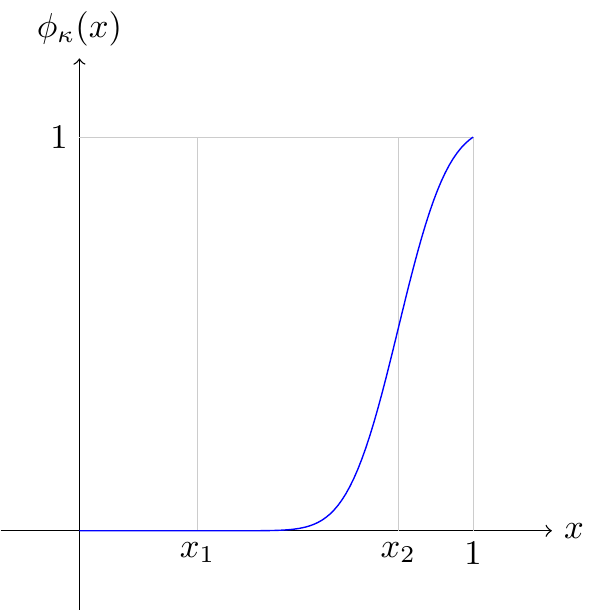}}}
\hfill
\subfloat[Fixation probability for coordination]{
	\parbox{0.45\linewidth}{\includegraphics[scale=1]{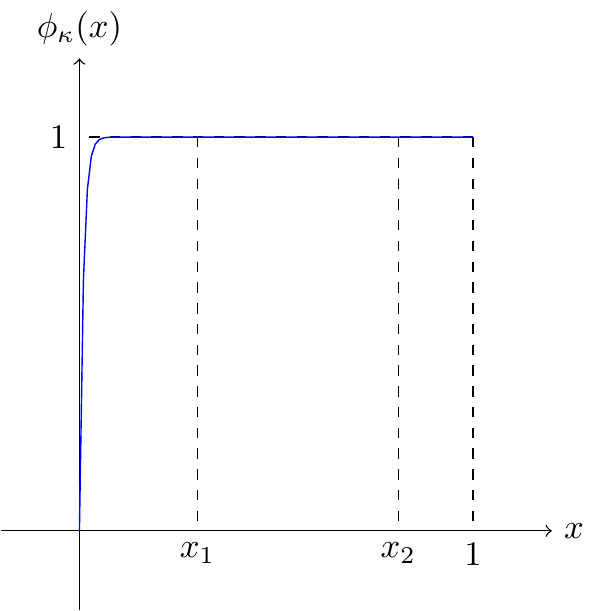}}}
\end{center}
\caption{Infinite versus finite population evolution features for case II. Again, we notice that  depending on the landscape of the potential, we can have either a behaviour that is coordination like (figures (\textsc{B}) and (\textsc{D})) or dominance like figures (\textsc{C}) and (\textsc{E})). In the latter case, the ESS will enhance the  of fixation of \A, which will be then almost certain. We will term this an evolutionary tunnelling.}
\end{figure}

\begin{remark}
Along the lines of remark~\ref{rmk:cc}, we point out that both equations \eqref{eqn:fix_two_ci} and  \eqref{eqn:fix_two_cii} can be seen as convex combinations of coordination with  either dominance by \B in the former or dominance by \A in the latter.
\end{remark}

\clearpage

\section{ESS in large  populations}

\label{sec:ess_lp}

\subsection{A continuous definition}
\label{ssec:ess_lpcont}

As already observed in the introduction, the definition of $\textsc{ESS}_N$ given by~\citet{NowakSasakiTaylorFudenberg,Nowak:06} has been widely accepted as a working definition for ESS for finite populations, although earlier definitions have also been given \citep{MaynardSmith1988,Schaffer1988,Neill2004}; see  also \citet{Ludwig:1975}. For the convenience of the reader, we recall this definition, but formulated in a compatible notation.

\begin{definition}[$\textsc{ESS}_N$]
\label{def:essn_nwk}
Consider a SBD process with a population size $N$, with $\Phi_N$ denoting the probability of fixation of $\A$. We say that   strategy  $\B$ is an $\textsc{ESS}_N$ if the following is satisfied:
\begin{enumerate}
\item $\Theta_N(1/N)<0$;
\item $\Phi_N(1/N)<1/N$;
\end{enumerate}
\end{definition}

For a family of regular SBD processes, definition~\ref{def:essn_nwk} can be recast in the continuous framework for fixation probabilities, provided $N$ is sufficient large as the next result shows:

\begin{theorem}
\label{thm:eqv_ess}
Consider a family of regular SBD processes with generalised log relative fitness $\Theta_N$  and let $\phi_\kappa$ be given by equation (\ref{eqn:asymp_frm}). Then, for sufficiently large $N$, \B is an $\textsc{ESS}_N$ if, and only if, we have that
\begin{enumerate}
\item $\phi_\kappa''(0)>0$;
\item  $\phi_\kappa(1/N)<1/N$.
\end{enumerate}
 \end{theorem}

\begin{proof}
Recall that for a family of regular SBD processes, we have that $\sfrac{\Theta_N}{\|\Theta_N\|_\infty}$ converges uniformly to $\theta$. Hence, for sufficiently large $N$ we have that $\Theta_N(1/N)<0$ if, and only if, $\theta(1/N)<0$. Moreover, continuity of $\theta$ and $N$ large enough implies that $\theta(0)<0$ if, and only if, $\theta(1/N)<0$.

On the other hand, notice that $\phi_\kappa$ satisfies
\[
\phi_\kappa''+\kappa^{-1}\theta\phi_\kappa'=0,\quad \phi_\kappa'>c'>0.
\]
Hence,
\[
\frac{\phi_\kappa''(0)}{\phi_\kappa'(0)}=-\kappa^{-1}\theta(0)>0,
\]
and then $\phi_\kappa''(0)>0$ as claimed.

Condition (2) is compatible with the error term given by \eqref{eqn:fit_is}, provided $N$ is sufficiently large.
\end{proof}

\begin{remark}
Conditions (1) and (2) in the theorem above   can be replaced by
\begin{enumerate}
\item $\phi''_\kappa(1/N)>0$;
\item $1/N-\phi_\kappa(1/N)>\dfrac{C}{N^\alpha}$ with $1<\alpha<2$. 
\end{enumerate}
In this case, $N$ has only to be large enough so that \eqref{eqn:fit_is} holds.  Notice that the conditions in Theorem~\ref{thm:eqv_ess} can be obtained from those above, by approximating $\phi_k''(1/N)$ by $\phi''_\kappa(0)$ in the first condition and by neglecting the term proportional to $N^{-\alpha}$ in the right hand side of  the second condition.
\end{remark}

Following the ideas in Theorem~\ref{thm:oneparameter}, we will allow for the possibility that $\theta$  may have also a dependence on a parameter vector $\lambda$, that we  will indicate by writing $\theta=\theta(x,\lambda)$. In this sense, the fixation is also dependent on $\lambda$, and will indicate this by writing $\phi_\kappa=\phi_{\kappa,\lambda}(x)$.

For evolutionary dynamics that satisfy the first condition for the existence of an ESS, the condition on the fixation probability of an invader required by (2) in Theorem~\ref{thm:eqv_ess} suggests the following definition:
\begin{definition}[Critical parameters]
The set of critical parameters for a population of size $N$ is defined as
\[
\mathfrak{C}_N=\left\{(\kappa,\lambda) \, \left| \phi_{\kappa,\lambda}\left(\frac{1}{N}\right)=\frac{1}{N}\right. \right\}.
\]
\end{definition}
Such a set can be thought of as defining a boundary, in parameter space, dividing $\textsc{ESS}_N$ strategies from non-$\textsc{ESS}_N$ strategies. For the particular case of linear fitness-differences --- i.e. in the weak-selection regime --- we have $\theta(x)=\gamma(x-\xs)$. Without loss of generality, we can use $\sigma^2=\kappa/\gamma$ and $\xs$ as parameters. For a fixed $\sigma$, one expects that for sufficient small $\xs$ the fixation probability of a invader will be equal to $\sfrac{1}{N}$, and hence that $(\sigma,\xs)$ will be in the critical parameter set. We shall term such an equilibrium $\xs$   a critical frequency at variance $\sigma$, or simply a critical frequency.

\subsection{$\textsc{ESS}_N$ in the quasi-neutral regime: generalisations of the one-third law}

Before tackling the general case,  we shall investigate the quasi-neutral case.   We begin by an asymptotic result for the continuous approximation given by \eqref{eqn:fit_is} that is valid in this regime:
\begin{theorem}
\label{thm:fix_qn}
Consider a regular family of SBD processes  in the quasi-neutral regime.  Then we have that
\begin{equation}
\label{eqn:fix_qn}
\phi_\kappa(x)=x+\kappa^{-1}\left[x\int_0^1(1-s)\theta(s)\,\rd s- \int_0^x(x-s)\theta(s)\,\rd s\right] + \kappa^{-2}x\RR(x;\kappa)+\bo{\kappa^{-3}},
\end{equation}
with $\RR=\bo{1}$ and smooth. Moreover,  its derivatives are also order one.
\end{theorem}
The proof of this theorem is given in Appendix~\ref{ap:qnr_proof}.

With this asymptotic result,  we can now obtain a general condition so that strategy  $\B$ is an $\textsc{ESS}_N$, when $N$ is large and we are in a restricted quasi-neutral regime. 

\begin{theorem}
\label{thm:gen_onethird}
Assume that we are in the quasi-neutral regime with $\kappa^{-1}=\lo{\sfrac{1}{N}}$,  and that we are in the coordination case. Then strategy $\B$ is an $\textsc{ESS}_N$ if, and only if,
\begin{enumerate}
\item $\theta(0)\ll -N^{-1}$
\item 
\[
\int_0^1(1-s)\theta(s)\,\rd s<\frac{\theta(0)}{2N}+\lo{\frac{1}{N}}.
\]
\end{enumerate}
\end{theorem}
\begin{proof}
Differentiating twice \eqref{eqn:fix_qn} yields
\[
\phi_\kappa''(x)=-\kappa^{-1}\theta(x) +\kappa^{-2}\left[x\RR''(x;\kappa)+2\RR'(x;\kappa)\right]+\bo{\kappa^{-3}},
\]
Hence $\phi_\kappa''(0)>0$ if, and only if, $\theta(0)<2\kappa^{-1}\RR'(0;\kappa) +\bo{\kappa^{-2}}$.

Now, we write
\[
\phi_\kappa(x)=x\left[1+\kappa^{-1}\int_0^1(1-s)\theta(s)\,\rd s\right]  - \kappa^{-1} G(x)+\kappa^{-2}x\RR(x;\kappa),\qquad G(x)=\int_0^x(x-s)\theta(s)\,\rd s.
\]
On noticing that $G(0)=G'(0)=0$, and that $G''(0)=\theta(0)$. We find that 
\begin{align*}
\phi_\kappa(1/N)-\frac{1}{N}&= \frac{1}{N}\left[1+\kappa^{-1}\int_0^1(1-s)\theta(s)\,\rd s +\kappa^{-2}\RR(\sfrac{1}{N},\kappa)\right]  -\kappa^{-1} G(\sfrac{1}{N})-\frac{1}{N}+\bo{\kappa^{-3}}\\
&=\frac{1}{\kappa N}\left[\int_0^1(1-s)\theta(s)\,\rd s - \frac{\theta(0)}{2N}+ \kappa^{-1}\RR(\sfrac{1}{N},\kappa)+\bo{\sfrac{1}{N^2}}+\bo{N\kappa^{-2}}\right].
\end{align*}
The last expression is negative if, and only if,
\[
\int_0^1(1-s)\theta(s)\,\rd s<\frac{\theta(0)}{2N}-\kappa^{-1}\RR(\sfrac{1}{N},\kappa) +  \bo{\frac{1}{N^2}}+\bo{N\kappa^{-2}}
=\frac{\theta(0)}{2N}+\lo{\frac{1}{N}}.
\]
\end{proof}

\begin{remark}
\label{rmk:sgol}
Naturally, for very large $N$, we will have  $\kappa^{-1}$ very small in the quasi-neutral regime. It is possible to show  that  $\RR(0,\kappa)=0$ and that $\RR$ is continuous. Hence, we can  rephrase  the conditions of the Theorem~\ref{thm:gen_onethird}  as that the left hand side of all inequalities should be negative. However, these conditions are sufficient, but not necessary in general. On the other hand, if we write $\Psi^{(\A,\B)}=1+\omega(N)\psi^{(\A,\B)}$, with $\omega(N)\ll1$, we can then identify  $\kappa^{-1}$ with  $N\omega(N)$ up to a  multiple that is independent of $N$. Thus, in this more conventional setting, we see that if $N$ is not large enough so that $\sfrac{1}{N}$ is not negligible, then we need the further assumption that $N^2\omega(N)\ll1$ instead of the usual assumption $N\omega(N)\ll1$. These observations go along  with the results in \citet{Wuetal2010}. 
\end{remark}
\begin{remark}
For large $N$, and for $\theta$ with a parameter dependence given by $\theta(x,\lambda)$, we see that Theorem~\ref{thm:gen_onethird} specifies that the critical parameter set can be approximated for small $\kappa$ by the following equation:
\[
\int_0^1(1-s)\theta(s,\lambda)\,\rd s=0.
\]
\end{remark}
In particular case where the  leading-order log-difference of fitness is linear, we recover a slightly generalised version of the calculation in \citet{Traulsenetal2006b} and we obtain the celebrated one-third law:
\begin{proposition}[one-third law]
Consider the case that $\theta$ is linear, i.e., $\theta(x)=\alpha(x-\xs)$, and assume that we are in the quasi-neutral regime. Then strategy $\B$ is an $\textsc{ESS}_N$ if, and only if,  $\xs>\sfrac{1}{3}+\bo{\sfrac{1}{N},\kappa^{-1}}$.
\end{proposition}
\begin{proof}
Since $\theta(0)=-\alpha\xs<0$, we compute 
\[
\int_0^1(1-s)\theta(s)\,\rd s=\alpha\left[\frac{1}{6}-\frac{\xs}{2}\right]=\frac{\alpha}{2}\left[\frac{1}{3}-\xs\right].
\]
and the later expression is negative if, and only if, $\xs>\sfrac{1}{3}$. The result then follows from Theorem~\ref{thm:gen_onethird}.
\end{proof}

The one-third law has been generalised to  $d$-player games for Moran processes \citep{Kurokawa:Ihara:2009,Gokhale:Traulsen:2010} and for a class of exchangeable processes \citep{Lessard:2011}. We shall now recover it as a special case of  Theorem~\ref{thm:gen_onethird}.

\begin{proposition}[Generalised one-third law for $d$-player games]\label{prop:dgames}
 Consider a $d$-player game, in a large population.  We have that  $\B$ is an $\textsc{ESS}_N$, if $a_0-b_0<0$, and if 
\[
\sum_{k=0}^{d-1}(d-k)a_k>\sum_{k=0}^{d-1}(d-k)b_k
\]
\end{proposition}

\begin{proof}
We begin by recalling equation~(\ref{eq:dgames}) that
\[
\theta(s)=\gamma\sum_{k=0}^{d-1}\binom{d-1}{k}s^k(1-s)^{d-1-k}(a_k-b_k).
\]
Then we have $a_0-b_0=\theta(0)<0$. On the other hand, we can compute
\begin{align*}
\int_0^1s^k(1-s)^{d-1-k}\,\rd s& = \frac{(d-k-1)!k!}{d!}=\frac{1}{d}\binom{d-1}{k}^{-1}\\
\int_0^1s^{k+1}(1-s)^{d-1-k}\,\rd s &= \frac{(d-k-1)!(k+1)!}{(d+1)!}=\frac{k+1}{d(d+1)}\binom{d-1}{k}^{-1}.
\end{align*}
Hence
\[
\gamma\int_0^1\theta(s)\,\rd s=\sum_{k=0}^{d-1}\frac{a_k-b_k}{d}
\quad\text{and}\quad
\gamma\int_0^1s\theta(s)\,\rd s=\sum_{k=0}^{d-1}\frac{k+1}{d(d+1)}(a_k-b_k).
\]
Then Theorem~\ref{thm:gen_onethird} implies that $\B$ will be an $\textsc{ESS}_N$ if 
\[
\int_0^1\theta(s)\,\rd s<\int_0^1s\theta(s)\,\rd s,
\]
which using the integrals above yields
\[
\sum_{k=0}^{d-1}\frac{a_k-b_k}{d}<\sum_{k=0}^{d-1}\frac{k+1}{d(d+1)}(a_k-b_k), 
\]
which yields the desired inequality.
\end{proof}

\begin{remark}
\label{rmk:otl}
The one-third law can be rephrased as saying that the critical frequency is 1/3 for very large variances. Notice also that Theorem~\ref{thm:gen_onethird} shows that the one-third law is not universal---as already noticed in \cite{Wuetal2010} in the context of linear fitness, but considering  the effects of higher order perturbations. As a matter of fact, it implies that it depends strongly on the fitness difference, given the integral nature of the condition obtained.
\end{remark}

\subsection{Risk dominance}

Let $\rho_\A$ be the probability of fixation of an individual $\A$ in a population with $N-1$ $\B$ individuals, and $\rho_\B$ be the probability of fixation of a $\B$ individual in a population with $N-1$ $\A$ individuals. This usually referred as the invasion probabilities.  
Following~\citet{Gokhale:Traulsen:2014} and references therein,  we shall say that $\A$ is risk dominant\footnote{This concept was introduced in the static formulation of game theory by~\citet{Harsanyi_Selten} as a Nash equilibrium refinement; its extension to EGT was made in~\citet{Kandori_Mailath_Rob}; see also~\citet{NowakSasakiTaylorFudenberg}.} over \B if $\rho_A>\rho_B$.

We start with the following result
\begin{theorem}
\label{thm:rd}
Assume weak-selection. Then strategy $\A$ in a finite large population  is risk dominant if and only if the fitness potential is negative at $x=1$, i.e., if $\F(1)<0$.
\end{theorem}
\begin{proof}

First, we observe that
\begin{align*}
\rho_\A&=\phi_N(\sfrac{1}{N})\\
&=\frac{1}{Nd_N}\left(1+\frac{\kappa_N^{-1}}{N}\exp(\kappa_N^{-1}\F(\chi^0_N))\right),&&0<\chi_N^0<\sfrac{1}{N}.
\end{align*}
Similarly, we have that
\begin{align*}
\rho_\B&=1-\phi_N(1-\sfrac{1}{N})\\
&=\frac{1}{Nd_N}\left(\exp(\kappa_N^{-1}\F(1))-\frac{\kappa_N^{-1}}{N}\exp(\kappa_N^{-1}\F(\chi_N^1))\right),&&1-\sfrac{1}{N}<\chi_N^1<1.
\end{align*}
Thus, we can write
\begin{align*}
\frac{\rho_\A}{\rho_\B}&=\exp\left(-\kappa_N^{-1}\F(1)\right)\frac{1+\sfrac{\kappa_N^{-1}}{N}\exp(\kappa_N^{-1}\F(\chi_N^0))}{1-\sfrac{\kappa_N^{-1}}{N}\exp\left(-\kappa_N^{-1}\left(\F(1)-\F(\chi_N^1)\right)\right)}\\
&=\exp\left(-\kappa_N^{-1}\F(1)\right)\frac{1+\sfrac{\kappa_N^{-1}}{N}\exp(\kappa_N^{-1}\F(\chi_N^0))}{1-\sfrac{\kappa_N^{-1}}{N}\exp(-\kappa_N^{-1}(1-\chi_N^1)\theta(\zeta_N))},&& \chi_N^1<\zeta_N<1.
\end{align*}
As $N\to\infty$,  we have that $\chi_N^0\to0$, $\chi_N^1\to1$ (and hence also that $\zeta_N\to1$). Further, since we are in the weak-selection regime,  we also have that $\sfrac{\kappa_N^{-1}}{N}\to0$ (and hence  also that $\kappa_N^{-1}(1-\chi_N^1)\to0$).  Therefore, we have that 
\[
\frac{\rho_\A}{\rho_\B}=\exp(-\kappa_N^{-1}\F(1))\left(1+\bo{\frac{\kappa_N^{-1}}{N}}\right),
\]
from which the result follows.

\end{proof}

We will now recover a result first proved in \citet{Kurokawa:Ihara:2009} (see also \cite{Gokhale:Traulsen:2010})

\begin{proposition}[Risk dominance for $d-$ person games]
In a large population, in the  weak-selection regime, we have that strategy $\A$ is risk-dominant if, and only if,
\[
\sum_{k=0}^{d-1}a_k>\sum_{k=0}^{d-1}b_k.
\]
\end{proposition}
\begin{proof}
The fitness function is given by equation~(\ref{eq:dgames}), and repeating integrals from the proof of proposition~\ref{prop:dgames}, we find
\[
\F(1)=-\frac{1}{d}\sum_{k=0}^{d-1}(a_k-b_k)\ ,
\]
from which the result follows.
\end{proof}

\begin{remark}
 If $d=2$, then $\F(1)=-\frac{\gamma}{2}\left(a+b-c-d)\right)$, and $\F(1)<0$ if and only if  $a+b>c+d$. 
In the coordination case we write $\theta(x)=\gamma(x-\xs)$ (possibly for a different constant $\gamma$) and
$\F(1)=-\gamma\left(\frac{1}{2}-\xs\right)$.
We conclude that $\A$ is risk-dominant if, and only if, $\xs<\sfrac{1}{2}$~\citep{Kandori_Mailath_Rob,NowakSasakiTaylorFudenberg}.
In the coexistence case we can now write $\theta(x)=-\gamma(\xs- s)$ (possibly for a different choice of $\gamma$). Then
$\F(1)=-\gamma\left(\xs-\frac{1}{2}\right)$,
and the reverse conclusion holds. Both cases can be seen in an unified manner: \A is risk-dominant if and only if $\F(1)<0$.
\end{remark}

\subsection{Critical frequency}

We now want to extend the study of the critical frequency to regimes outside the quasi-neutral one. For convenience of presentation, we focus in the case of weak-selection, with leading-order linear log-difference fitness. In this case, the asymptotic solution given by \eqref{eqn:fp_asymp:cd} is an exact solution of \eqref{eqn:fit_is} for all $\kappa$. Thus,   we can use it to investigate the critical frequency numerically. The result of such an  investigation is shown in Figure~\ref{fig:cf }.

\begin{figure}[htbp]
\subfloat[Large variance end of the critical frequency curve.  \label{fig:cf_a}]{\includegraphics[scale=0.35]{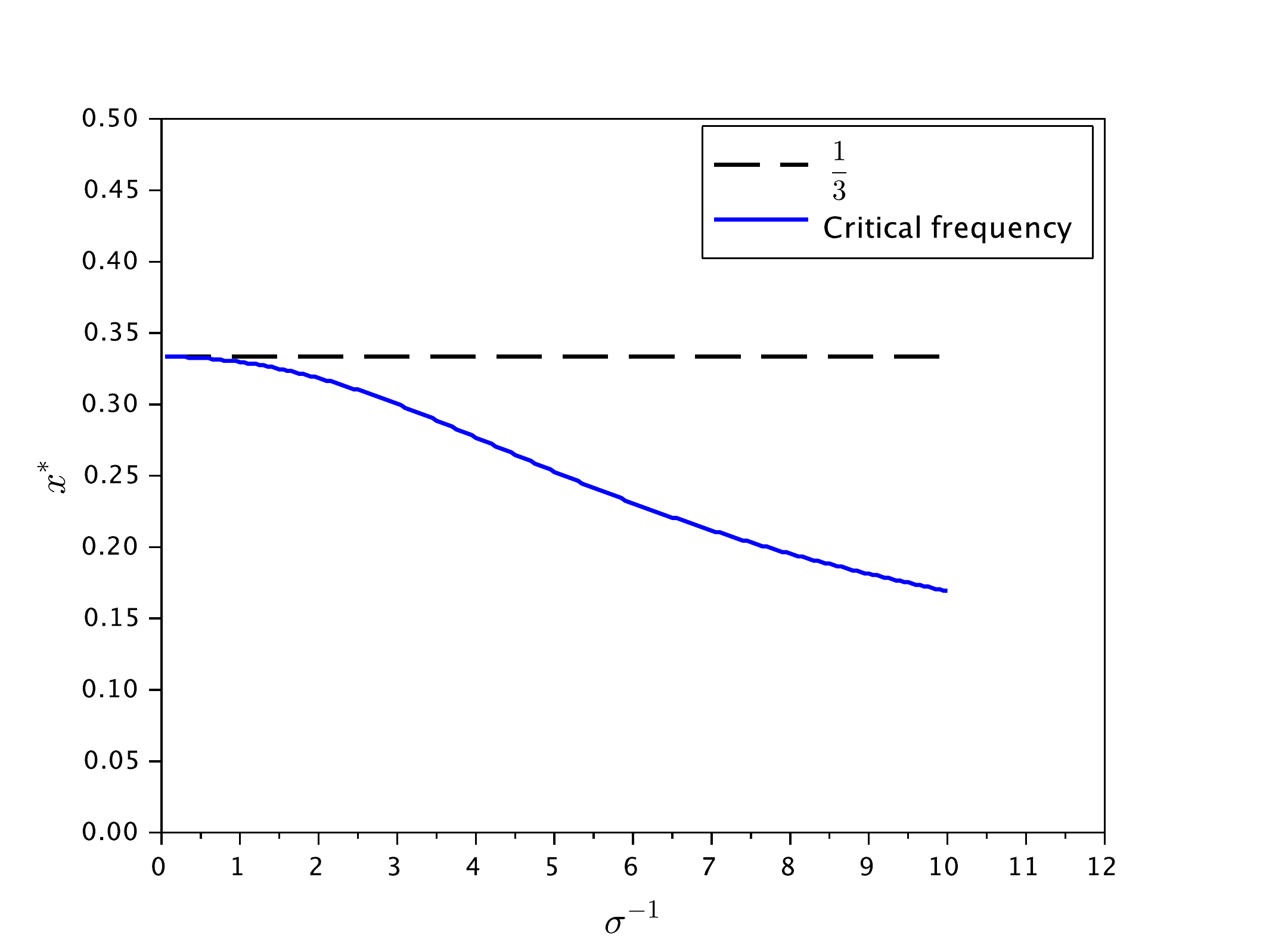}}
\subfloat[Small variance end of the critical frequency curve.  \label{fig:cf_b}]{\includegraphics[scale=0.35]{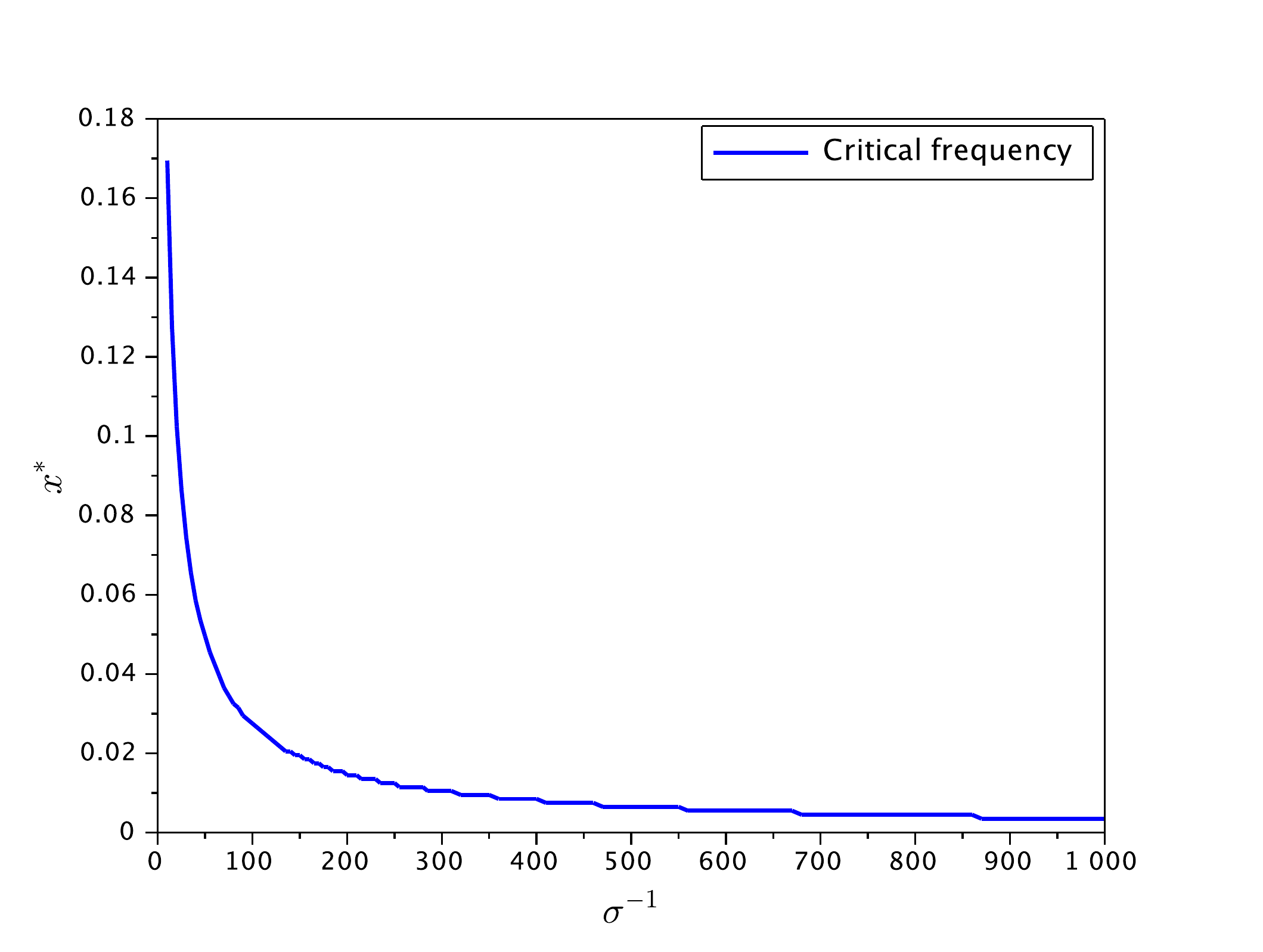}}
\caption{Critical curves obtained using the fixation probability  given by \eqref{eqn:fp_asymp:cd} and  $N=10000$. Notice that for very large variance the critical frequency approaches $\sfrac{1}{3}$ as predicted by the one-third law. Additionally, notice that for order one variance, the critical frequency is still significant, being above 0.15. As the variance decreases, the critical frequency diminishes, but even at variances small as $10^{-6}$, the critical frequency is about 0.0035. For a variance of order $10^{-4}$, which would be the typical variance in the absence of weak-selection,  the critical frequency is about 0.03. \label{fig:cf }}
\end{figure}

These results suggest a number of consequences as noted in the following:
\begin{enumerate}
\item As observed in Remark~\ref{rmk:otl}, the one-third law states that  the critical frequency for evolutionary dynamics in the quasi-neutral regime is $\sfrac{1}{3}$. The numerical results shown in Figure~\ref{fig:cf_a} indicate that in the case of an evolutionary dynamic  in the balanced regime,  the critical frequency is below one-third, but still at a significant level.
\item For evolutionary dynamics in the selection-driven regime, the picture in Figure~\ref{fig:cf_b} is universal---as a consequence of the asymptotic results in Section~\ref{sec:sdr_asymp}---and suggests that even when $\kappa_N^{-1}$ is moderately small , the critical frequency can still be an issue for the selection of a Nash equilibrium in the finite population case.
\end{enumerate}

\section{Discussion}
\label{sec:discuss}

We presented a discussion about fixation for large, but finite populations within a continuous perspective. This was possible by the derivation of a continuous approximation that is valid in a large range of evolutionary regimes, at the expense of requiring further  regularity of the logarithm of the relative fitness. With this derivation, we can identify a number of evolutionary regimes with and without weak-selection assumptions. In the latter case, we obtain a selection-driven regime. In the former case, the possibilities are larger: we can also have a selection-driven regime, but also a balanced evolution or a quasi-neutral regime.  In particular, it seems that there is some confusion in the literature regarding the concepts of weak-selection  and quasi-neutral regimes. As we have seen, the latter implies the former, but not the other way round. In particular, \citet{NowakSasakiTaylorFudenberg} do need a quasi-neutral regime to derive the one-third law, as seen in section~\ref{sec:ess_lp}, and most population-genetics literature defines weak-selection as  quasi-neutrality---cf. \citet{Ewens_2004,WildTraulsen2007}.

The derivation carried out in Section~\ref{sec:setup} suggests that in the absence of weak-selection, the infinite limit of the Moran process is not the usual replicator dynamics but what might be called the ``log-replicator'' equation
\begin{equation}
\label{eqn:rd_wws}
\dot{x}=x(1-x)\log\left[\frac{\psi^\A(x)}{\psi^\B(x)}\right]
\end{equation}
Notice that \eqref{eqn:rd_wws} is topologically conjugated to the replicator dynamics, i.e. \eqref{eqn:rd_wws} has the same equilibria that the replicator, and the equilibria have the same stability properties; in addition the sense of time is also preserved. In other words, the qualitative  picture does not change in the infinite population setting. Notice, nevertheless, that such invariance will generally only hold for two types. For three or more types, the dynamics of the logarithm-replicator equation can be qualitatively very different from the replicator dynamics. In addition, there is quite a number of differences in the finite population case as discussed below.

We also presented asymptotic formulae for the fixation of probability, in the selection driven regime --- with or  without weak-selection. We consider mainly the case of at most a single interior equilibrium. For the dominance case, we recover the results of~\citet{Kimura,AntalScheuring_2006} in a unified way. For the coexistence case, however, the presented formulas seem to be novel, and show that typically one might expect coexistence dynamics from infinite population emerging from some kind of dominance at finite populations. For the case of coordination, we derive a formula that slightly extends the ones obtained by \citet{MobiliaAssaf2010,AssafMobilia2010} and also allows the equilibria to be close of the endpoints. 

Using these derived approximations, we showed what we called the near one-half law: with weak-selection, and when the leading order of the logarithm of the relative fitness is linear---and this is the case that leads to the replicator dynamics, cf.~\citet{ChalubSouza09b,ChalubSouza_JMB}---we show that, unless the equilibria is within a layer of $\bo{\kappa}$ of $\sfrac{1}{2}$, we have either dominance by $\A$ or by $\B$, depending on the position of the equilibrium relative to $\sfrac{1}{2}$, i.e. if its slightly larger or smaller, respectively. In this sense, while the equilibrium gives some indication about the fixation pattern, there is a phase transition (in the limit $N\to\infty$)
effect leading to dominance of one of the types. In the  small layer where the dominance switching takes place, and that we term the coexistence layer, both types have significant nonzero fixation probabilities. In the absence of weak-selection, the fixation pattern is even more varied. We present an example of two pay-off matrices, with Nash equilibrium located at $x=\sfrac{3}{4}$, but which display opposite fixation patterns in the finite population case. In particular, we show that adding a constant to the pay-off matrix can dramatically change the fixation behaviour. 

We also briefly study the case of multiple equilibria. Here the results are also new, and show that if we have a coordination equilibrium followed by a coexistence one, then the latter can be a blockage for the former. More precisely, depending on the landscape of the corresponding potential, one can have either a coordination like  or a dominance by \B like fixation pattern.  On the other way round, we see that the ESS can now act as tunnel, and allow for the evolution to bypass the evolutionary barrier imposed by the coordination equilibrium. Namely, depending once again on the potential landscape, one can have either a coordination like or dominance by \A fixation pattern.

Finally, we study the existence of  ESS in finite populations. Using the diffusive approximation, we formulate an ESS definition for large populations, and show that it is equivalent to the so-called $\textsc{ESS}_N$ condition. We also introduce the concept of critical parameters, for studying the possible cases of ESS in parameter space. For linear fitness differences, this amounts to describe what we call the critical frequency for a given variance. Following that, we obtain an asymptotic approximation for the fixation in the quasi-neutral regime, and use this to obtain a rather general condition for a strategy that opposes invasions of a mutant to be an ESS. The condition is an integral one, and shows that the whole behaviour of the logarithm of the relative fitness is important, rather than just local information as in the selection-driven regime. In the special case of leading order linear fitness differences, we recover the celebrated one-third law~\citep{NowakSasakiTaylorFudenberg}, while for $d$-player games we recover the so-called generalised one-third law \citep{Gokhale:Traulsen:2010,Lessard:2011}.  We then proceed further to study what happens outside the quasi-neutral regime. We then focus on the case of linear fitness for convenience of presentation, and provides a numerical study of what we call the critical curve. This shows that even outside the quasi-neutral regime, the critical frequency can be significantly non zero, and this might have significant implications in  understanding some aspects of evolutionary dynamics. As an additional result, we give a continuous formulation of a risk-dominant strategy and show that it recovers the corresponding definition for $d$-player games \citep{Kurokawa:Ihara:2009}.

%
%


\appendix

\section{Derivation of \eqref{eqn:fit_is}}
\label{ap:deriv:fit_is}

First, observe that
\[
 \prod_{r\in[\sfrac{1}{N},s]_{N}}\frac{\Delta^-_N(r)}{\Delta^+_N(r)}=\exp\left(\sum_{r\in[\sfrac{1}{N},s]_{N}}\log\left(\frac{\Delta^-_N(r)}{\Delta^+_N(r)}\right)\right)=\exp\left(-\sum_{r\in[\sfrac{1}{N},s]_{N}}\Theta_N(r)\right).
\]

Now we observe that
\begin{align*}
\sum_{r\in[\sfrac{1}{N},s]_{N}}\Theta_N(r)&=\|\Theta_N\|_\infty\sum_{r\in[\sfrac{1}{N},s]_{N}}\frac{\Theta_N(r)}{\|\Theta_N\|_\infty}\\
&=\|\Theta_N\|_\infty\sum_{r\in[\sfrac{1}{N},s]_{N}}\left(\frac{\Theta_N(r)}{\|\Theta_N\|_\infty}-\theta(r)\right)+\kappa^{-1}_N\sum_{r\in[\sfrac{1}{N},s]_{N}}\frac{\theta(r)}{N}.
\end{align*}
The last sum can be interpreted as a Riemann sum in two different ways: either as a right sum, or as a midpoint sum.  The classical error bounds for the Riemann sums \citep{Atkinson1989,SB2002}, are as follows:
\[
\left\|\frac{\theta(x_0+\sfrac{1}{N})}{N}-\int_{x_0}^{x_0+\sfrac{1}{N}}\theta(r)\,\rd r\right\|\leq \frac{\theta'(c)}{2N^2},\quad c\in (x_0,x_0+\sfrac{1}{N}),
\]
for the  right rule, and 
\[
\left\|\frac{\theta(x_0+\sfrac{1}{N})}{N}-\int_{x_0-\delta_N}^{x_0+\delta_N}\theta(r)\,\rd r\right\|\leq 
\frac{\theta''(c)}{24N^3},\quad c\in (x_0-\delta_N,x_0+\delta_N),
\]
for the midpoint rule, where
\[
\delta_N=\frac{1}{2N}.
\]
 These bounds yield the following  simple bounds:

\[
\left\|\sum_{r\in[\sfrac{1}{N},s]_{N}}\frac{\theta(r)}{N}-\int_0^s\theta(r)\,\rd r\right\|_\infty\leq \frac{\|\theta'\|_\infty}{2N};
\]
for the former, whereas, in the latter, we have
\[
\left\|\sum_{r\in[\sfrac{1}{N},s]_{N}}\frac{\theta(r)}{N}-\int_{\delta_N}^{s+\delta_N}\theta(r)\,\rd r\right\|_\infty\leq 
\frac{\|\theta''\|_\infty}{24N^2}.
\]
In addition, we also have that
\[
\left\|\sum_{r\in[\sfrac{1}{N},s]_{N}}\left(\frac{\Theta_N(r)}{\|\Theta_N\|_\infty}-\theta(r)\right)\right\|_\infty\leq N\epsilon_N.
\]
Combining these two results, we find that 
\[
\sum_{r\in\left[\sfrac{1}{N},s-\sfrac{1}{N}\right]_N}\Theta_N(r)=\kappa^{-1}_N\int_{\delta_N}^{s-\delta_N}\theta(r)\,\rd r + \kappa^{-1}_N\epsilon_NR_1^0(s)+\kappa^{-1}_NN^{-2}R_2^0(s), 
\]
where
\[
\kappa_N^{-1}\epsilon_NR_1^0(s)=\|\Theta_N\|_\infty\sum_{r\in[\sfrac{1}{N},s]_{N}}\left(\frac{\Theta_N(r)}{\|\Theta_N\|_\infty}-\theta(r)\right)
\quad
\text{and}
\quad
\kappa_N^{-1}N^{-2}R_2^0(s)=\sum_{r\in[\sfrac{1}{N},s]_{N}}\frac{\theta(r)}{N}-\int_{\delta_N}^{s+\delta_N}\theta(r)\,\rd r,
\]
and hence we have that $\|R_1^0\|_\infty$ and $\|R_2^0\|_\infty$ are bounded uniformly in $N$.
Recalling that
\[
\F(s)=-\int_0^s\theta(r)\,\rd r,
\]
we then  have that
\begin{align*}
\sum_{r\in\left[\sfrac{1}{N},s-\sfrac{1}{N}\right]_N}\Theta_N(r)&= -\kappa^{-1}_N\F(s-\delta_N) +\underbrace{\kappa^{-1}_N\F(\delta_N)}_{\C_N} + \underbrace{\kappa^{-1}_N\epsilon_NR_1(s)+\kappa^{-1}_NN^{-2}R_2(s)}_{\G_N(s)} \\ 
&=-\kappa^{-1}_N\F(s-\delta_N)+\C_N+\G_N(s).
\end{align*}
Thus
\begin{align*}
\sum_{s\in[\sfrac{1}{N},x]_N}\prod_{r\in[\sfrac{1}{N},s-\sfrac{1}{N}]_N}\frac{\Delta^-_N(r)}{\Delta^+_N(r)}&=\exp(-\C_N)\sum_{s\in[\sfrac{1}{N},x]_N}\exp\left(\kappa_N^{-1}\F(s-\delta_N)\right)\exp(-\G_N(s))\\
&=\exp(-\C_N)\exp(\kappa_N^{-1}\F(\bar{s}))\sum_{s\in[\sfrac{1}{N},x]_N}\exp\left(\kappa_N^{-1}\HH(s-\delta_N)\right)\exp(-\G_N(s)).
\end{align*}
where $\bar{s}$ is any point where the global maximum of $\F$ is attained, and $\HH(s)=\F(s)-\F(\bar{s})$.

Let
\[
\Upsilon_N=\kappa_N^{-1}\max\left\{\|R_1^0\|_\infty\epsilon_N,\|R_2^0\|_\infty N^{-2}\right\}.
\]
Since $\|\G_N\|_\infty\leq\Upsilon_N$, we can find $E_N(x)$, with $\|E_N(x)\|_\infty\leq\Upsilon_N$, such that
\[
\sum_{s\in[\sfrac{1}{N},x]_N}\prod_{r\in[\sfrac{1}{N},s-\sfrac{1}{N}]_N}\frac{\Delta^-_N(r)}{\Delta^+_N(r)}=\exp(-\C_N)\exp(\kappa_N^{-1}\F(\bar{s}))\sum_{s=[\sfrac{1}{N},x]_N}\exp\left(\kappa_N^{-1}\HH(s-\delta_N)\right)(1+E_N(x)).
\]
Therefore, we have
\begin{align*}
&\sum_{s\in[\sfrac{1}{N},x]_N}\prod_{r\in[\sfrac{1}{N},s-2\delta_N]_N}\frac{\Delta^-_N(r)}{\Delta^+_N(r)}=\\
&\qquad=N\exp(-C_N)\exp\left(\kappa_N^{-1}\F(\bar{s})\right)\sum_{s\in[\sfrac{1}{N},x]_N}\frac{1}{N}\exp\left(\kappa_N^{-1}\HH(s-\delta_N)\right)(1+E_N(x))\\
&\qquad=N\exp(-\C_N)\exp\left(\kappa_N^{-1}\F(\bar{s})\right)\left[\int_{\delta_N}^{x+\delta_N}\exp\left(\kappa_N^{-1}\HH(s-\delta_N)\right)\,\rd s + R_N(x)\right](1+E_N(x))\\
&\qquad=N\exp(-\C_N)\exp\left(\kappa_N^{-1}\F(\bar{s})\right)\left[\int_{0}^{x}\exp\left(\kappa_N^{-1}\HH(s)\right)\,\rd s + R_N(x)\right](1+E_N(x))
\end{align*}
with
\[
R_N(x)=\frac{\kappa_N^{-1}}{24N^3}\sum_{j=1}^m\exp\left(\kappa_N^{-1}\HH(\bs_j)\right)\left(\kappa_N^{-1}\theta^2(\bs_j)-\theta'(\bs_j)\right), 
\]
where $\bs_j\in (\sfrac{1}{2N}+\sfrac{(j-1)}{N},\sfrac{1}{2N}+\sfrac{j}{N})$, and with $m=\lfloor xN\rfloor$.

If we write
\[
I_N(x)=\int_{0}^{x}\exp(\kappa_N^{-1}\HH(s))\,\rd s,
\]
then, by combining all the previous calculations,  we obtain the following approximation:
\begin{align*}
\Phi_N(x)&=\frac{I_N(x)+R_N(x)}{I_N(1)+R_N(1)}\times\frac{1+E_N(x)}{1+E_N(1)}
=\frac{\sfrac{I_N(x)}{I_N(1)}+\sfrac{R_N(x)}{I_N(1)}}{1+\sfrac{R_N(1)}{I_N(1)}}\times\frac{1+E_N(x)}{1+E_N(1)}\\
&=\frac{\phi_N(x)+\QQ_N(x)}{1+\D_N}\times\frac{1+E_N(x)}{1+E_N(1)}
\end{align*}
where
\[
\D_N=\frac{R_N(1)}{I_N(1)}, \quad \QQ_N(x)=\frac{R_N(x)}{I_N(1)},
\]
In addition, we have also used that
\[
\phi_{N}(x)=\frac{I_N(x)}{I_N(1)}=d_N^{-1}\int_0^x\exp\left(\kappa_N^{-1}\F(s)\right)\,\rd s,\quad d_N=\int_0^1\exp\left(\kappa_N^{-1}\F(s)\right)\,\rd s.  
\]

If $\kappa_N^{-1}$ has a finite limit as $N\to\infty$, then 
\[
\left| \kappa_N^{-1}\theta^2(x)-\theta'(x)  \right|\leq C.
\]
Hence
\begin{equation*}
|R_N(x)|\leq \frac{Cm}{N^3}\leq \frac{C}{N^2}.
\end{equation*}
For $x>\sfrac{1}{N}$, since $I_N(x)>I_N(\sfrac{1}{N})$, we have
\[
|\QQ_N(x)|\leq C\frac{\sfrac{1}{N^2}}{\exp(\kappa_N^{-1}\HH(0))\sfrac{1}{N}+\bo{\sfrac{1}{N^2}}}\leq \frac{C}{N}.
\]
This proves Equation~\eqref{eqn:fit_is_un}.

For the remaining results, we first observe that an asymptotic argument using Watson's lemma along the lines discussed in  Section~\ref{sec:sdr_asymp} and Appendix~\ref{ap:asymp_proof} yields
\[
\left|I_N(1)\right|=
\left\{
\begin{array}{lr}
C\kappa_N+\bo{\kappa_N^2},&\text{if }\F\text{ is a boundary potential};\\
C\kappa_N^{1/2}+\bo{\kappa_N},&\text{if }\F\text{ is an interior  potential}.
\end{array}
\right.
\]
To bound $R_N(x)$ we will need the following Lemma:
\begin{lemma}
If $\kappa_N^{-1}$ is not bounded as $N\to\infty$, then we have
\[
\left|R_N(x)\right|\leq C
\left\{
\begin{array}{lr}
\frac{\kappa_N^{-1}}{N^2},&\F\text{ is a boundary potential;}\\
\frac{\kappa_N^{-1/2}}{N^2},&\F \text{ is an interior potential}.
\end{array}
\right.
\]
\end{lemma}
Thus, we immediately obtain combining the asymptotic estimations together with the Lemma that:
\[
\left|\D_N\right|,\left|\QQ_N(x)\right|\leq 
\left\{
\begin{array}{lr}
C\frac{\kappa_N^{-2}}{N^2},&\text{if }\F\text{ is a boundary potential};\\
C\frac{\kappa_N^{-1}}{N^2},&\text{if }\F\text{ is an interior  potential}.
\end{array}
\right.
\]
Notice also that if $\Phi_N(x)$ is exponentially small then we must have $\F(s)<0$, $s\in (0,x)$. Hence we also have that $\phi_N(x)$ is exponentially small and also $R_N(x)$ is exponentially small.  This proves Equation~\eqref{eqn:asymp_frm}.

To prove Equation~\eqref{eqn:fit_is_nb}, notice that if $\kappa_N^{-1}$ is not bounded, and if $\phi_N(x)=\sfrac{1}{N}$, then 
\[
I_N(x)=\frac{1}{N}\left\{
\begin{array}{lr}
\kappa_N,& \text{if }\F\text{ is a boundary potential};\\
\kappa_N^{1/2},&\text{if }\F\text{ is an interior  potential}.
\end{array}
\right.
\]
Hence
\[
|\tilde{\QQ}_N(x)|\leq \left\{
\begin{array}{lr}
\frac{\kappa_N^{-2}}{N},& \text{if }\F\text{ is a boundary potential};\\
\frac{\kappa_N^{-1}}{N},&\text{if }\F\text{ is an interior  potential}.
\end{array}
\right.
\]
Thus the continuous approximation can correctly identify the neutral boundary, provided $\kappa_N=\bo{N^{\alpha}}$, with $\alpha<\sfrac{1}{2}$, if $\F$ is a boundary potential, or that evolution is in the moderate selection regime, if $\F$ is an interior potential. 

\begin{proof}[Proof of the Lemma]
We now observe that 
\begin{align*}
R_N(x)&=\frac{\kappa_N^{-1}}{24N^3}\sum_{j=1}^m\exp\left(\kappa_N^{-1}\HH(\bs_j)\right)\left(\kappa_N^{-1}\theta^2(\bs_j)-\theta'(\bs_j)\right)\\
&=\frac{\kappa_N^{-1}}{24N^2}\left[\int_0^x\exp\left(\kappa_N^{-1}\HH(s)\right)\left(\kappa_N^{-1}\theta^2(s)-\theta'(s)\right)\,\rd s + \RR_N(x)\right]\\
&=\frac{\kappa_N^{-1}}{24N^2}\left[\exp\left(\kappa_N^{-1}\HH(0)\right)\theta(0)-\exp\left(\kappa_N^{-1}\HH(x)\right)\theta(x)+  \RR_N(x)\right],
\end{align*}
where
\[
\RR_N(x)=\frac{\kappa_N^{-1}}{2N^2}\sum_{j=1}^m\exp\left(\kappa_N^{-1}\HH(\hs_j)\right)\left(-\kappa_N^{-1}\theta^3(\hs_j)+3\theta(\hs_j)\theta'(\hs_j)-\kappa_N\theta''(\hs_j)\right),
\]
with $\hs_j\in (\sfrac{1}{2N}+\sfrac{(j-1)}{N},\sfrac{1}{2N}+\sfrac{j}{N})$.

If $\kappa_N^{-1}$ is bounded then, we can bound 
\[
\left\|\exp\left(\kappa_N^{-1}\HH(x)\right)\left(-\kappa_N^{-1}\theta^3(x)+3\theta(x)\theta'(x)-\kappa_N\theta''(x)\right)\right\|_\infty <C,\quad \text{independent of }N.
\]
Hence, we have that
\[
|\RR_N(x)|\leq \frac{C}{N}.
\]
Otherwise, if $\kappa_N^{-1}$ is not bounded, we have the following bounds:
\[
\left|\RR_N\right(x)|\leq 
\left\{
\begin{array}{lr}
C\frac{\kappa_N^{-1}}{N^2},&\text{if }\F\text{ is a boundary potential};\\
C\frac{\kappa_N^{-1/2}}{N},&\text{if }\F\text{ is an interior  potential}.
\end{array}
\right.
\]
Indeed, if $\F$ is a boundary potential, then we have either that $\HH(0)=0$ or that $\HH(1)=0$. We we will treat the former, the latter being similar. In this case, let $\st$ be the smallest interior global minimum, if it exists, or $\st=1$ if there is no interior global minimum. Then there exists  $\tK$ and 
$0<\tx=\sfrac{m}{N}<\st$, such that
\[
\tK s\leq-\HH(s),\quad s\in [0,\tx].
\]
Then
\begin{align*}
\left|\sum_{j=0}^m\exp\left(\kappa_N^{-1}\HH(\hs_j)\right)\left(-\kappa_N^{-1}\theta^3(\hs_j)+3\theta(\hs_j)\theta'(\hs_j)-\kappa_N\theta''(\hs_j)\right)\right|
&\leq \kappa_N^{-1}M'\sum_{j=0}^m\exp(-\kappa_N^{-1}\tK\bs_j)\\
&\leq \kappa_N^{-1}M'\int_0^\infty \exp(-\kappa_N^{-1}\tK s)\,\rd s\\
&= \tC_{-1},\quad \tC_{-1}=\ord(1).
\end{align*}

If $\F$ is an interior potential, let us write $\xs$ for any of its interior maxima. Recall that, in this case, we have $\theta(\xs)=0$ and $\theta'(\xs)>0$. Let
\[
J(x)=\exp\left(\kappa_N^{-1}\HH(x)\right)\left[-\kappa_N^{-1}\theta^3(x)+3\theta(x)\theta'(x)-\kappa_N\theta''(x)\right],\quad x\in [0,1].
\]
We claim that $\|J\|_\infty=\bo{\kappa_N^{1/2}}$.  To see this, let 
\[
\tilde{J}(x)=\exp\left(\kappa_N^{-1}\HH(x)\right)\left[-\kappa_N^{-1}\theta^3(x)+3\theta(x)\theta'(x)\right]
\]
and compute
\[
\tilde{J}'(x)=\exp\left(\kappa_N^{-1}\HH(x)\right)\left[\kappa_N^{-2}\theta^4(x)-6\kappa_N^{-1}\theta^2(x)\theta'(x)+3\left({\theta'}^2(x)+\theta(x)\theta''(x)\right)\right].
\]
Then $\tilde{J}'(x)=0$ is equivalent to
\[
\theta^4(x)-6\kappa_N\theta^2(x)\theta'(x)+3\kappa_N^{2}\left({\theta'}^2(x)+\theta(x)\theta''(x)\right)=0.
\]
Firstly,  we observe that we  are only interested in solutions close to $\xs$, since $\tilde{J}$ is exponentially small otherwise. An analysis of the magnitude of the terms in the previous equation, suggests that if $\bx$ is a solution, then $|\theta(\bx)|=\bo{\kappa_N^{1/2}}$. Since this problem is a regular perturbation --- but where we can not apply the implicit function theorem --- we  write
 \[
 \bx=\xs+\kappa_N^{1/2}x_1+\bo{\kappa_N}
 \]
 which yields the following equation for $x_1$:
 \[
\left( \theta'(\xs)\right)^4x_1^4-6\left(\theta'(\xs)\right)^3x_1^2+3\left(\theta'(\xs)\right)^2=0.
 \]
The solutions are
 \[
 x_1=\pm\left(\theta'(\xs)\right)^{-1/2}\sqrt{3\pm\sqrt{6}}.
 \]
 and it can be easily verified that two of these solutions correspond to local minima of $\tilde{J}$ that are close to $\xs$, while the two other correspond  to local maxima. In any case, a direct computation yields 

 \begin{align*}
 \HH(\bx)&=\HH(\xs)+\HH'(\xs)\kappa_N^{1/2}x_1+\frac{\kappa_N}{2}\HH''(\xs)x_1^2+\bo{\kappa_N^{3/2}}\\
 &=0-\theta(\xs)\kappa_N^{1/2}x_1-\frac{\kappa_N}{2}\theta'(\xs)x_1^2+\bo{\kappa_N^{3/2}}\\
 &=-\frac{\kappa_N}{2}\left(3\pm\sqrt{6}\right) +\bo{\kappa_N^{3/2}}.
 \end{align*}
 Hence
 \[
 \exp\left(\kappa_N^{-1}\HH(\bx)\right)=\exp\left(-\frac{\left(3\pm\sqrt{6}\right)}{2}+\bo{\kappa_N^{1/2}}\right).
 \]
 Also
 \begin{align*}
 -\kappa_N^{-1}\theta^3(\bx)+3\theta(\bx)\theta'(\bx)&=-\kappa_N^{-1}\left(\left(\theta'(\xs)\right)^3\kappa_N^{3/2}x_1^3+\bo{\kappa_N^2}\right)+3\kappa_N^{1/2}x_1\theta'(\xs)^2+\bo{\kappa_N}\\
 &=\kappa_N^{1/2}\left[-\left(\theta'(\xs)\right)^3x_1^3+3x_1\theta'(\xs)^2\right]+\bo{\kappa_N}\\
 &=\kappa_N^{1/2}\left[3\left(\theta'(\xs)\right)^{3/2}\left(3\pm\sqrt{6}\right)^{1/2}-\left(\theta'(\xs)\right)^{3/2} \left(3\pm\sqrt{6}\right)^{3/2}    \right]+\bo{\kappa_N}\\
 &=\kappa_N^{1/2}\left(\theta'(\xs)\right)^{3/2}\left[3\left(3\pm\sqrt{6}\right)^{1/2}- \left(3\pm\sqrt{6}\right)^{3/2}    \right]+\bo{\kappa_N},\\
 &=\mp\sqrt{6}\left(3\pm\sqrt{6}\right)^{1/2}\theta'(\xs)^{3/2}\kappa_N^{1/2}+\bo{\kappa_N}.
 \end{align*}
 Therefore, we have 
 \color{black}
 \[
 |\tilde{J}(\bx)|=\bo{\kappa_N^{1/2}}
 \]
 and hence we conclude that 
\[
\|\tilde{J}\|_\infty\leq C\kappa_N^{1/2}.
\]
Since we can easily bound
\[
\|\exp\left(\kappa_N^{-1}\HH\right)\kappa_N\theta''\|_\infty<C\kappa_N,
\]
we can conclude that
\[
\|J\|_\infty\leq \|\tilde{J}\|_\infty + \|\exp\left(\kappa_N^{-1}\HH\right)\kappa_N\theta''\|_\infty\leq C_1\kappa_N^{1/2}+C_2\kappa_N\leq C\kappa_N^{1/2}.  
\]
This yields the bounds on $\RR_N$. We now proceed to  estimate
\[
R_N(x)=\frac{\kappa_N^{-1}}{24N^2}\left[\exp\left(\kappa_N^{-1}\HH(0)\right)\theta(0)-\exp\left(\kappa_N^{-1}\HH(x)\right)\theta(x)+  \RR_N(x)\right].
\]
If  $\F$ is a boundary potential, we have that either $\HH(0)=0$ or $\HH(1)=0$, and hence
\[
|R_N(x)|\leq C\frac{\kappa_N^{-1}}{24N^2}.
\]
If $\F$ is an interior potential, let us write
\[
L(x)=\exp\left(\kappa_N^{-1}\HH(0)\right)\theta(0)-\exp\left(\kappa_N^{-1}\HH(x)\right)\theta(x).
\]
Then
\[
L'(x)=\exp\left(\kappa_N^{-1}\HH(x)\right)\left[\kappa_N^{-1}\theta^2(x)-\theta'(x)\right]
\]
Now notice that a solution $L'(x)=0$ is given by $\bx=\xs+\kappa_N^{1/2} x_1$. Direct substitution in $L$ then yields
\[
\|L\|_\infty\leq C\kappa_N^{1/2}.
\]
Hence, we obtain that 
\[
|R_N(x)|\leq C\frac{\kappa_N^{-1/2}}{N^2}.
\]
\end{proof}

\section{Proof of the asymptotic results}
\label{ap:asymp_proof}
Write equation~\eqref{eqn:asymp_frm} as
\[
\phi_\kappa(x)=\frac{I(x)}{I(1)},\qquad I(x)=\int_{0}^{x}\exp\left(\kappa^{-1}\F(s)\right)\,\rd s.
\]

\subsection{Dominance}
For dominance of $\A$, we have $\theta(x)>0$ in unit interval, and hence the argument of the exponential has a maximum at $s=0$. Let $s=\kappa z $. Then, using Laplace's method~\citep{Hinch1991,BenderOrszag1999}, we find 
\begin{align*}
I(x)&=\kappa\int_0^{x/\kappa}\exp\left(\kappa^{-1}\F(\kappa z )\right)\,\rd z\\
&=\frac{\kappa}{\theta(0)}\left(1-\exp\left(\frac{-\theta(0)x}{\kappa}\right)\right)+\bo{\kappa^2}.
\end{align*}
Hence, we have
\begin{equation}
\label{eqn:dombya:ap}
\phi_\kappa(x)=1-\exp\left(-\frac{\theta(0)x}{\kappa}\right)+\bo{\kappa}.
\end{equation}
For dominance of $\B$, recall that we have $\theta(x)<0$ throughout  $[0,1]$. Hence. the argument in exponential will then have a maximum at $s=1$. Thus, we write $s=1-\kappa z$ and, analogously as before, we find
\[
I(x)=-\frac{\kappa\exp(\kappa^{-1}\F(1))}{\theta(1)}\left[\exp\left(\frac{\theta(1)(1-x)}{\kappa}\right)-\exp\left(\frac{\theta(1)}{\kappa}\right)+\bo{\kappa}\right].
\]
Hence, we find
\begin{equation}
\label{eqn:dombyb:ap}
\phi_\kappa(x)=\exp\left(\frac{\theta(1)(1-x)}{\kappa}\right)+\bo{\kappa}.
\end{equation}
\subsection{Coexistence}

In the case of coexistence, the fitness potential has a minimum at $x=\xs$; hence we have no contribution from the interior. On the other hand,  $\theta$ is positive near $s=0$ and negative near $s=1$. Hence, the argument in the exponential has a maximum at both $s=0$ and $s=1$. Hence combining the previous calculations we find
\begin{equation}
\label{eqn:i:coex}
I(x)=\frac{\kappa}{\theta(0)}\left(1-\exp\left(\frac{-\theta(0)x}{\kappa}\right)\right)-\frac{\kappa\exp(\kappa^{-1}\F(1)}{\theta(1)}\left[\exp\left(\frac{\theta(1)(1-x)}{\kappa}\right)-\exp\left(\frac{\theta(1)}{\kappa}\right)+\bo{\kappa}\right].
\end{equation}
If $\F(1)\ll-\kappa$, then the second term of \eqref{eqn:i:coex} is exponentially small, and hence we obtain once again \eqref{eqn:dombya:ap}. On the other hand, if $\F(1)\gg\kappa$, the second term is then exponentially large, and in this case we obtain \eqref{eqn:dombyb:ap}.

Otherwise, if $\F(1)\sim\kappa$, let
\[
C=\exp(\F(1)/\kappa)\quad\text{and}\quad \gamma=\frac{|\theta(1)|}{\theta(0)}.
\]
We now have that \eqref{eqn:i:coex} becomes 
\[
I(x)=\frac{\kappa}{\theta(0)}\left(1-\exp\left(\frac{-\theta(0)x}{\kappa}\right)\right)+\frac{\kappa C}{|\theta(1)|}\exp\left(\frac{\theta(1)(1-x)}{\kappa}\right)+\bo{\kappa^2}.
\]
 Also, we have
\[
I(1)=\kappa\frac{|\theta(1)|+\theta(0)C}{\theta(0)|\theta(1)|} +\bo{\kappa^2}.
\]
Therefore,
\begin{equation}
\label{eqn:coex:ap}
\phi_\kappa(x)=\frac{\gamma}{\gamma+C}\left(1-\exp\left(\frac{-\theta(0)x}{\kappa}\right)\right) + \frac{C}{\gamma+C}\exp\left(\frac{\theta(1)(1-x)}{\kappa}\right)+ \bo{\kappa}.
\end{equation}
\subsection{Coordination}

For coordination, we have that the fitness potential has a maximum at $x=\xs$. Hence, we write
\[
\F(s)=\F(\xs)-\theta'(\xs)\frac{(s-\xs)^2}{2} +\bo{(s-\xs)^3}.
\]
Then, if we write
\[
z=\sqrt{\frac{\theta'(\xs)}{\kappa}}(s-\xs),
\]
Hence we have that
\begin{align*}
I(x)&=\sqrt{\frac{\kappa}{\theta'(\xs)}}\exp\left(\kappa^{-1}\F(\xs)\right)\int_{-\sqrt{\frac{\theta'(\xs)}{\kappa}}\xs}^{\sqrt{\frac{\theta'(\xs)}{\kappa}}(x-\xs)}\exp\left(-\frac{z^2}{2}\right)\,\rd z +\bo{\kappa}\\
&=\sqrt{\frac{2\pi\kappa}{\theta'(\xs)}}\exp\left(\kappa^{-1}\F(\xs)\right)\left[\frac{1}{\sqrt{2\pi}}\int_{-\infty}^{\sqrt{\frac{\theta'(\xs)}{\kappa}}(x-\xs)}\exp\left(-\frac{z^2}{2}\right)\,\rd z
+\int_{-\infty}^{-\sqrt{\frac{\theta'(\xs)}{\kappa}}\xs}\exp\left(-\frac{z^2}{2}\right)\,\rd z\right]+\bo{\kappa}\\
&=\sqrt{\frac{2\pi\kappa}{\theta'(\xs)}}\exp\left(\kappa^{-1}\F(\xs)\right)\left[\N\left(\sqrt{\frac{\theta'(\xs)}{\kappa}}(x-\xs)\right) - \N\left(-\sqrt{\frac{\theta'(\xs)}{\kappa}}\xs\right)\right]+\bo{\kappa}.
\end{align*}
Therefore, we find that
\[
\phi_\kappa(x)=\frac{\N\left(\sqrt{\frac{\theta'(\xs)}{\kappa}}(x-\xs)\right) - \N\left(-\sqrt{\frac{\theta'(\xs)}{\kappa}}xs\right)}{\N\left(\sqrt{\frac{\theta'(\xs)}{\kappa}}(1-\xs)\right) - \N\left(-\sqrt{\frac{\theta'(\xs)}{\kappa}}\xs\right)} +\bo{\sqrt{\kappa}}.
\]
\section{Proof of Theorem~\ref{thm:fix_qn}}
\label{ap:qnr_proof}
As before, we write
\[
\phi_\kappa(x)=\frac{I_\kappa(x)}{I_\kappa(1)},\qquad I_\kappa(x)=\int_0^x\exp(\F(s)/\kappa)\,\rd s.
\]
Write 
\[
\exp\left(\F(s)/\kappa\right)=1+\kappa^{-1}\F(s)+\kappa^{-2}\F(s)^2\sum_{l=0}^\infty\frac{1}{\kappa^{l}(l+2)!}\F(s)^{l}
\]
and integrate to obtain:
\[
I_\kappa(x)=x+\kappa^{-1}\int_0^x\F(s)\,\rd s+\kappa^{-2}H(x;\kappa),
\]
where
\[
H(x;\kappa)=\int_0^x\F(s)^2\sum_{l=0}^\infty\frac{1}{\kappa^l(l+2)!}\F(s)^l\,\rd s,
\]
which is order one. Notice that this will be also true for its derivatives.

Since $H$ is $C^3$ and $H(0,\kappa)=\partial_xH(0,\kappa)=\partial_x^2H(0,\kappa)=0$, we can invoke Hadamard Lemma, cf. \citet{BruceGiblin92}, and write
\[
H(x;\kappa)=x^3\HHH(x;\kappa),
\]
with $\HH$ being $C^2$.

Hence, we have
\[
\phi_\kappa(x)=x-\kappa^{-1}\left[x\int_0^1\F(s)\,\rd s-\int_0^x\F(s)\,\rd s\right] + \kappa^{-2}R(x;\kappa)+ \bo{\kappa^{-3}},
\]
where, 
\[
R(x;\kappa)=x\left[\int_0^1\F(s)\,\rd s\right]^2-x\HHH(1;\kappa)+x^3\HHH(x;\kappa)-\int_0^x\F(s)\,\rd s\int_0^1\F(s)\,\rd s.
\]
Since $R(0;\kappa)=0$, a further application of Hadamard Lemma yields
\[
R(x;\kappa)=x\RRR(x;\kappa),
\]
with $\RRR$ being $C^2$.

Finally, observe that integration by parts imply that
\[
\int_0^x\F(s)\,\rd s= -\int_0^x\F(s)\frac{\rd}{\rd s}\left[(x-s)\right]\,\rd s=\int_0^x(x-s)\F'(s)\,\rd s=-\int_0^x(x-s)\theta(s)\,\rd s.
\]

\bibliographystyle{jmb}
\bibliography{note}

\end{document}